\documentclass{wlscirep}
\pdfoutput=1
\usepackage{graphicx}
\usepackage{subfigure}
\newtheorem{theorem}{Theorem}
\newtheorem{proof}{Proof}
\newtheorem{remark}{Remark}

\title{Generalized exceptional quantum walk search}

\author[1,2]{Meng Li}
\author[1,3*]{Yun Shang}

\affil[1]{Institute of Mathematics, Academy of Mathematics and Systems Science, Chinese Academy of Sciences, Beijing 100190, China}
\affil[2] {School of Mathematical Sciences, University of Chinese Academy of Sciences, Beijing 100049, China}
\affil[3]{NCMIS, MDIS, Academy of Mathematics and Systems Science, Chinese Academy of Sciences, Beijing, 100190, China}

\affil[*]{shangyun602@163.com}


\begin{abstract}
We mainly study exceptional configuration for coined quantum walk search. For searching on a two-dimensional grid by AKR algorithm, we find some new classes of exceptional configurations that cannot be found by the AKR algorithm effectively and the known diagonal configuration can be regarded as its special case. Meanwhile, we give two modified quantum walk models that can improve the success probability in the exceptional configurations by numerical simulation. Furthermore, we introduce the concept of generalized exceptional configuration and consider search by quantum walk on a cycle with Grover coin. We find that the most natural coin combination model $(G,-)$, where $G$ is a Grover diffusion transformation, is a generalized exceptional configuration when just searching one marked vertex on the cycle. In the end, we find generalized exceptional configuration has a different evolution of quantum coherence from exceptional configuration. These extend largely the range of exceptional configuration of quantum walk search in some sense.
\end{abstract}

\begin{document}

\maketitle

\section{Introduction}
Quantum walks, as a quantum analogue of classical random walks, have been a useful model in designing quantum algorithms for a variety of problems~\cite{ambainis2005coins, childs2003exponential, matrixproduct, elementdistinctness, yang2018quantum, Portugal2018qwsearch}. In many of those applications, quantum walks are used as a tool for search problems~\cite{matrixproduct, elementdistinctness, Portugal2018qwsearch}. Search is one of the major problems in computer science.
And it is also a central task in the field of quantum algorithms.
Grover search~\cite{groversearch} is one of the well-known quantum search algorithms. It could search an unordered database of $N$ items in $O(\sqrt{N})$ time, which yields a quadratic speedup compared with the corresponding classical algorithm.
It was used recursively by Aaronson and Ambainis~\cite{aaronson2003quantum} for searching in grids. Shenvi, Kempe and Whaley~\cite{shenvi2003quantum} pointed out the potential of discrete time quantum walks with respect to searching problems and designed a quantum walk based simulation of Grover search.

Many search problems may be regarded as the problem of finding one or many marked vertices
from search space.
Ambainis, Kempe and Rivosh \cite{ambainis2005coins} designed a $(G,-I)$-Type quantum walk to search one or two marked vertices on a $\sqrt{N} \times \sqrt{N}$ grid by using the discrete time quantum walks, where $G$ represents the Grover's diffusion transformation, which is called AKR algorithm. The success probability of finding the marked vertices can reach $O(\frac{1}{\log N})$.
Szegedy~\cite{szegedy2004quantum} developed a theory of quantum walk based search algorithms by quantizing classical Markov chains. He generalized the number of marked vertex into arbitrary number.

Different from the performance of classical search algorithm, the quantum search algorithm may become harder when the number of marked vertices increases.
As for the AKR algorithm ~\cite{ambainis2005coins}, Ambainis and Rivosh~\cite{ambainis2008exceptional} showed an exceptional configuration of marked vertices where $N$ marked vertices are just on the diagonal of the $N\times N$ grid. This system only evolves by  flipping signs and stays in a uniform probability distribution for all time.
Thus, the success probability does not grow over time.
Wong and Santos~\cite{ThomasRM17} proved that one-dimensional cycle with any arrangement of marked vertices is an exceptional configuration under Szegedy's quantum walk~\cite{szegedy2004quantum} or its equivalent coined quantum walk, and constructed the higher-dimensional generalization, such as the two-dimensional grid with a marked diagonal.

In addition, Pr{\=u}sis et al.~\cite{Pr2016stationary} introduced the stationary states, i.e. 1-eigenvectors, of the quantum walk search operator for given graphs and configurations of marked vertices.
Nahimovs et al.~\cite{nahimovs2015exceptional, nahimovs2017adjacent} presented some configurations on two-dimensional grid and even general graph that have corresponding stationary states.
Recently, when the configurations have stationary states, the probability of finding marked vertices has also been analyzed~\cite{khadiev2018probability, glos2019upperbounds}.
Obviously, when the initial state (equal superposition state) happens to be the stationary state, the system remains unchanged all the time. Therefore, it is also an exceptional case which deserves our attentions.

In general, exploring the exceptional configuration or stationary state for a search algorithm can help us understand the algorithm more profoundly.
On the one hand, it helps us to refine the application scope of quantum algorithms and thus achieve the efficiency of quantum algorithms.
On the other hand, it may inspire us to design better search algorithm or explore other unknown applications based on the properties of exceptional configuration or stationary state.
For example, the algorithm for perfect matching in bipartite graph detecting using the property of stationary state~\cite{khadiev2018probability} has been considered.
Here, we will introduce a more general conception ``generalized exceptional configuration" for search problems (the initial state is always the equal superposition state), which means that, for an arrangement of marked vertices, the probability of success is always the same as their probability in the initial state no matter how many steps we take.
This concept includes more exceptional cases.
The exceptional configuration shown in Ref. \cite{ambainis2008exceptional} require the system must evolve by flipping signs and thus the success probability unchanged, so it belongs to generalized exceptional configuration.
And when the initial state happens to be the stationary state, the corresponding configuration is also a generalized exceptional configuration.

In this paper, we first give some new classes of exceptional configurations for AKR algorithms and provide two modified quantum walk models that can solve this exceptional search in some sense.
Then we introduce the concept of ``generalized exceptional configuration" by extending the range of exceptional configuration of quantum walk search and we find the success probability of the most natural $(G,-)$-Type coined quantum walk search algorithm will not grow over time when searching one marked vertex on the cycle.
Then we give a distinguish condition for the above two exceptional concepts by calculating their coherence.

The paper is organized as follows.
In Section 2, we introduce the models of quantum walk on the $N \times N$ grid and $N$-cycle, and provide the definitions of exceptional configuration and generalized exceptional configuration.
In Section 3, we discuss the exceptional cases of quantum walk search on the grid.
In Section 3.1, we present some new classes of exceptional configurations for AKR algorithm and we find that the known diagonal configuration is a special case. In Section 3.2, we give two modified quantum walk models that can improve the success probability in the exceptional configurations.
In Section 4, we consider the exceptional cases of quantum walk search on the cycle.
In Section 4.1, the generalized exceptional configuration about searching on the one-dimensional cycle based on the $(G,-)$-Type is discussed, where $G$ is Grover diffusion transformation. In Section 4.2, we give the strict difference between exceptional configuration and generalized exceptional configuration by an example-$(G,H)$-Type quantum walk, where $H$ is Hadamard matrix.
In Section 5,  we analyze the dynamics of quantum coherence in the exceptional configuration and generalized exceptional configuration which have been discussed above.
Then, we end with a summary in Section 6.

\section{Preliminaries}
For quantum walks on an arbitrary graph $G(V,E)$, they can be described by the repeated applications of an unitary operator $U=S\cdot C$ that acts on a Hilbert space $\mathcal{H}=\mathcal{H}^{c}\otimes\mathcal{H}^{p}$ where
$\mathcal{H}^{p}$ is the position Hilbert space spanned by the vertex $v\in V$ and $\mathcal{H}^{c}$ is the coin Hilbert space spanned by the edge $e\in E$.
$S$ is a shift operator that applies on the combined space, and $C=\widetilde{C}\otimes I$ is the coin transformation, in which $\widetilde{C}$ applies on the coin space.
So the basis states of the walker are $|j,v\rangle$ for $j\in \{0, 1 ,\cdots, d_{v}-1\}$, $v\in V$, and $d_{v}$ is the degree of vertex $v$. The state of the quantum walks is given by
\begin{equation}
|\psi\rangle=\sum_{j,v}a_{j,v}|j,v\rangle.
\end{equation}
If it comes to the search problem on a given graph, there are some marked vertices. The coin operator is $C=C_{0}\otimes (I-\Sigma_{v}|v\rangle\langle v|)+C_{1}\otimes \Sigma_{v}|v\rangle\langle v|$, where $|v\rangle$ represents the marked vertex. According to the choice of $C_{0}$ and $C_{1}$, we then call the corresponding quantum walk $(C_{0},C_{1})$-Type.
And the initial state of the whole system is the equal superposition state.
Next, we introduce two specific situations, the coined quantum walk on $N \times N$ grid and $N$-cycle, which will be discussed in detail later.

\textbf{Case 1 (coined quantum walk on the $N \times N$ grid)}
For quantum walks on a two-dimensional grid of size $N\times N$, the position space $\mathcal{H}^{p}$ is spanned by the basis states $|i,j\rangle$ for $i,j\in\{0,1,\ldots,N-1\}$, and the coin space $\mathcal{H}^{c}$ is spanned by basis states $\{|\uparrow\rangle,|\downarrow\rangle,|\leftarrow\rangle,|\rightarrow\rangle\}$.
So the basis states are $|d,i,j\rangle$ for $i,j\in \{0,1,\ldots,N-1\}$, $d\in \{\uparrow,\downarrow,\leftarrow,\rightarrow\}$, and the state of the system is given by
\begin{equation}
|\psi\rangle=\sum_{i,j}(a_{\uparrow,i,j}|\uparrow,i,j\rangle+a_{\downarrow,i,j}|\downarrow,i,j\rangle+a_{\leftarrow,i,j}|\leftarrow,i,j\rangle+a_{\rightarrow,i,j}|\rightarrow,i,j\rangle).
\end{equation}
For the marked vertices, Ambainis et al.~\cite{ambainis2005coins} apply $-I$ as the coin flip operator $C_{1}$.
In many search algorithm based on the discrete time coined quantum walk~\cite{ambainis2005coins, groversearch, shenvi2003quantum, nahimovs2015exceptional}, the nature choice is to apply Grover's diffusion transformation on the unmarked vertices as the coin flip transformation $C_{0}$.
The Grover's diffusion transformation is $G=2|\gamma\rangle\langle\gamma|-I$, where $|\gamma\rangle$ is the equal superposition state of the coin space. For the problem of searching on the $N\times N$ grid,
\begin{equation}
G=\frac{1}{2}
\left(
\begin{array}{cccc}
 -1 &  1  &  1  &  1 \\
 1  &  -1 &  1  &  1 \\
 1  &  1  &  -1 &  1 \\
 1  &  1  &  1  & -1
\end{array}
\right).
\end{equation}
The shift transformation $S$ is
\begin{equation}
\begin{array}{ccc}
 |\uparrow,i,j \rangle     &  \longmapsto  &  |\downarrow,i,j-1 \rangle  \\
 |\downarrow,i,j \rangle   &  \longmapsto  &  |\uparrow,i,j+1 \rangle    \\
 |\leftarrow,i,j \rangle   &  \longmapsto  &  |\rightarrow,i-1,j \rangle \\
 |\rightarrow,i,j\rangle  &  \longmapsto  &  |\leftarrow,i+1,j \rangle
\end{array}
\end{equation}
And the initial state is
\begin{equation} \label{gridinitial}
|\psi(0)\rangle=\frac{1}{\sqrt{4N^{2}}}\sum_{i,j}(|\uparrow,i,j \rangle +|\downarrow,i,j \rangle +|\leftarrow,i,j \rangle +|\rightarrow,i,j \rangle).
\end{equation}
Notice that the grid has periodic boundary conditions, i.e., the operations are modulo $N$. For example, the shift transformation $S$ maps $|\leftarrow,0,j \rangle$ to $|\rightarrow,N-1,j \rangle$, where $j\in \{0,1,\ldots,N-1\}$.

\textbf{Case 2 (coined quantum walk on the $N$-cycle)}
For quantum walks on one-dimensional cycle with $N$ vertices, the coin space $\mathcal{H}^{c}$ is spanned by the basis states $\{|0\rangle, |1\rangle\}$ which represent the direction of the walker moves, and the position space $\mathcal{H}^{p}$ is spanned by the basis states $\{|0\rangle, |1\rangle, \cdots, |N-1\rangle\}$.
The state of the system is described as
\begin{equation}
|\psi\rangle=\sum_{d=0}^{1}\sum_{x=0}^{N-1}a_{d,x}|d,x\rangle.
\end{equation}
In the two dimensional coin space, the Grover's diffusion transformation equals to the Pauli matrix $X$,
\begin{equation}
X=
\left(
  \begin{array}{cc}
    0 & 1 \\
    1 & 0
  \end{array}
\right).
\end{equation}
Also, the general unitary coin operator can be written as
\begin{equation}
Q=
\left(
\begin{array}{cc}
\sqrt{\rho} &  \sqrt{1-\rho}e^{i\theta} \\
\sqrt{1-\rho}e^{i\phi}  &  -\sqrt{\rho}e^{i(\theta +\phi)}
\end{array}
\right),
\end{equation}
where $\rho \in [0,1]$, $\theta \in [0,\pi]$ and $\phi \in [0,2\pi]$.
In particular, when $\rho=\frac{1}{2}$ and $\theta=\phi=0$, the matrix $Q$ is just the well-known Hadamard matrix $H$.
For the shift transformation S, there are two different operators which are often used in general.
\begin{equation}
S=|0\rangle\langle0|\otimes\sum_{x=0}^{N-1}|x-1\rangle\langle x|+|1\rangle\langle1|\otimes\sum_{x=0}^{N-1}|x+1\rangle\langle x|, \label{moving}
\end{equation}
\begin{equation}
S=|1\rangle\langle0|\otimes\sum_{x=0}^{N-1}|x-1\rangle\langle x|+|0\rangle\langle1|\otimes\sum_{x=0}^{N-1}|x+1\rangle\langle x|. \label{flipflop}
\end{equation}
The first is the moving shift operator, where a particle hops and keeps the direction the same.
And the second is the flip-flop shift operator, where the particle hops and changes the direction.
Note that the addition operation is defined modulo N.
The initial state of the search problem on the $N$-cycle is:
\begin{equation}
|\psi(0)\rangle=\frac{1}{\sqrt{2N}}\sum_{d=0}^{1}\sum_{x=0}^{N-1}|d,x\rangle.
\end{equation}

Now, we will list two concepts with respect to the success probability of finding marked vertices.

\textbf{Definition 1 (exceptional configuration)} For the search problem on a given undirected graph  under $(C_{0},C_{1})$-Type coined quantum walk, if there is an arrangement of marked vertices such that the system only evolves by flipping signs and thus the probability of success is always the same as their probability in the initial state over time, we call this kind of arrangement as an exceptional configuration corresponding to $(C_{0}, C_{1})$-Type quantum walk.

\textbf{Definition 2 (generalized exceptional configuration)} For the search problem on a given undirected graph under $(C_{0},C_{1})$-Type coined quantum walk, if there is an arrangement of marked vertices such that the success probability is always the same as their probability in the initial state over time, we call this kind of arrangement as a generalized exceptional configuration corresponding to $(C_{0},C_{1})$-Type quantum walk.

In this paper,
on one hand, we find some new classes of exceptional configurations of the marked vertices on the $N\times N$ grid, in which the known diagonal configuration presented in~\cite{ambainis2008exceptional} is just a special case.
To search the marked vertices with a higher probability in these exceptional configurations at some special steps, we provide two kinds of modified quantum walk models which can solve this problem successfully in terms of numerical simulation.
On the other hand, we further discuss the exceptional configurations and generalized exceptional configurations of the search problem on one-dimensional $N$-cycle with one marked vertex under some different types of coined quantum walks in detail.
Furthermore, we consider the evolution of quantum coherence in the two cases discussed above.

\section{Quantum walk search on the grid}
In this section, we introduce some new classes of exceptional configurations and two modified quantum walk models about searching on the two-dimensional $N \times N$ grid.

\subsection{The extension of the diagonal configuration}

So far now, the known exceptional configuration in the AKR algorithm is the diagonal case~\cite{ambainis2008exceptional}. However, we find that there exists a wider class of exceptional configurations apart from the diagonal case. In fact, when the coordinate of the marked vertex $(i,j)$ satisfies either $j-i=\alpha$ or $i+j=\alpha$ where $\alpha$ is an
arbitrary integer, the success probability does not change over time.
Moreover, the diagonal configuration is a special case where $j-i$ is equal to zero.
To understand the exceptional configurations, we give the example on $5\times 5$ grid shown in Figure \ref{55grid}, in which (a)-(e) represents the case when $j-i=\alpha$.
Note that the addition and subtraction in the position space are all calculated under the modulus $N$.
We prove it in the following theorem.

\begin{figure}[htb]
 \centering
 \subfigure[]{\label{}
 \includegraphics[width=3cm]{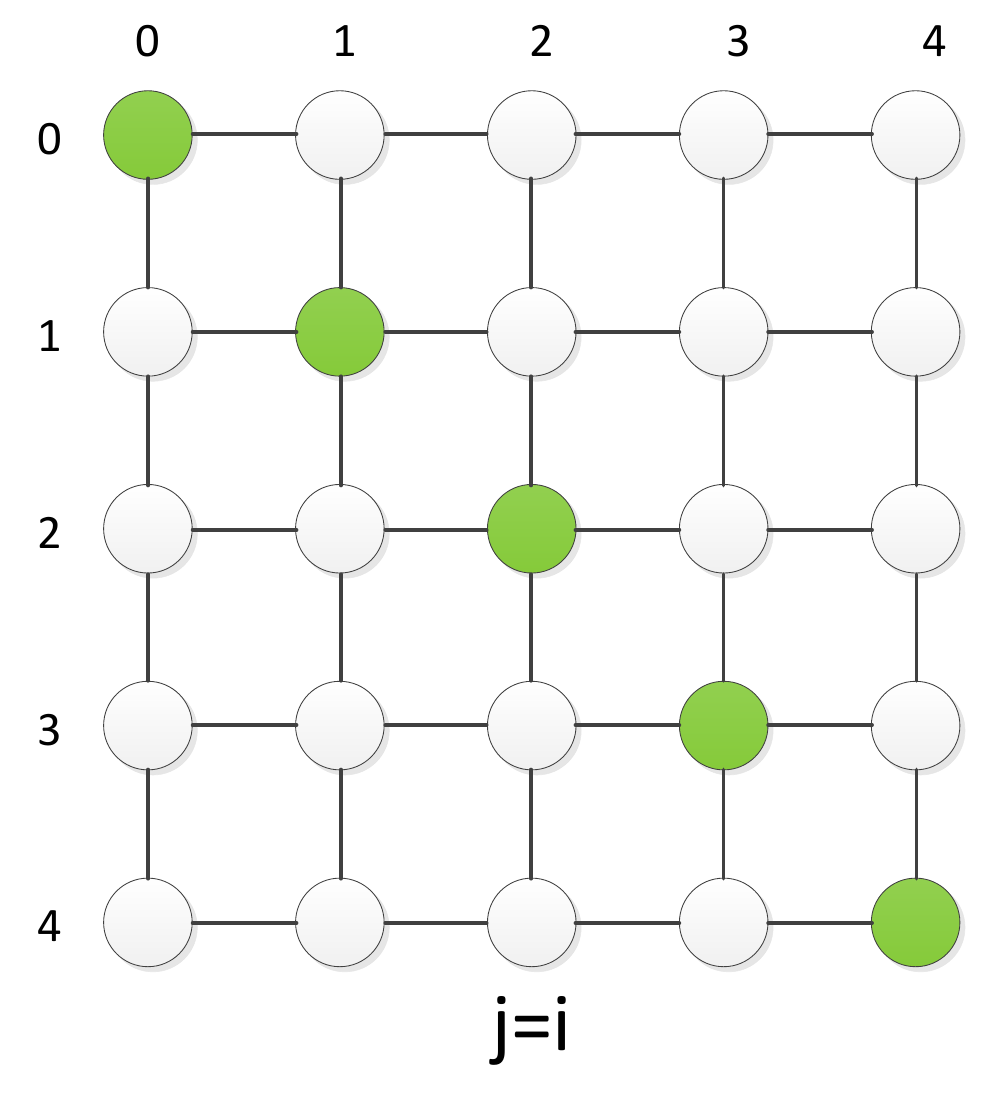}}
 \subfigure[]{\label{}
 \includegraphics[width=3cm]{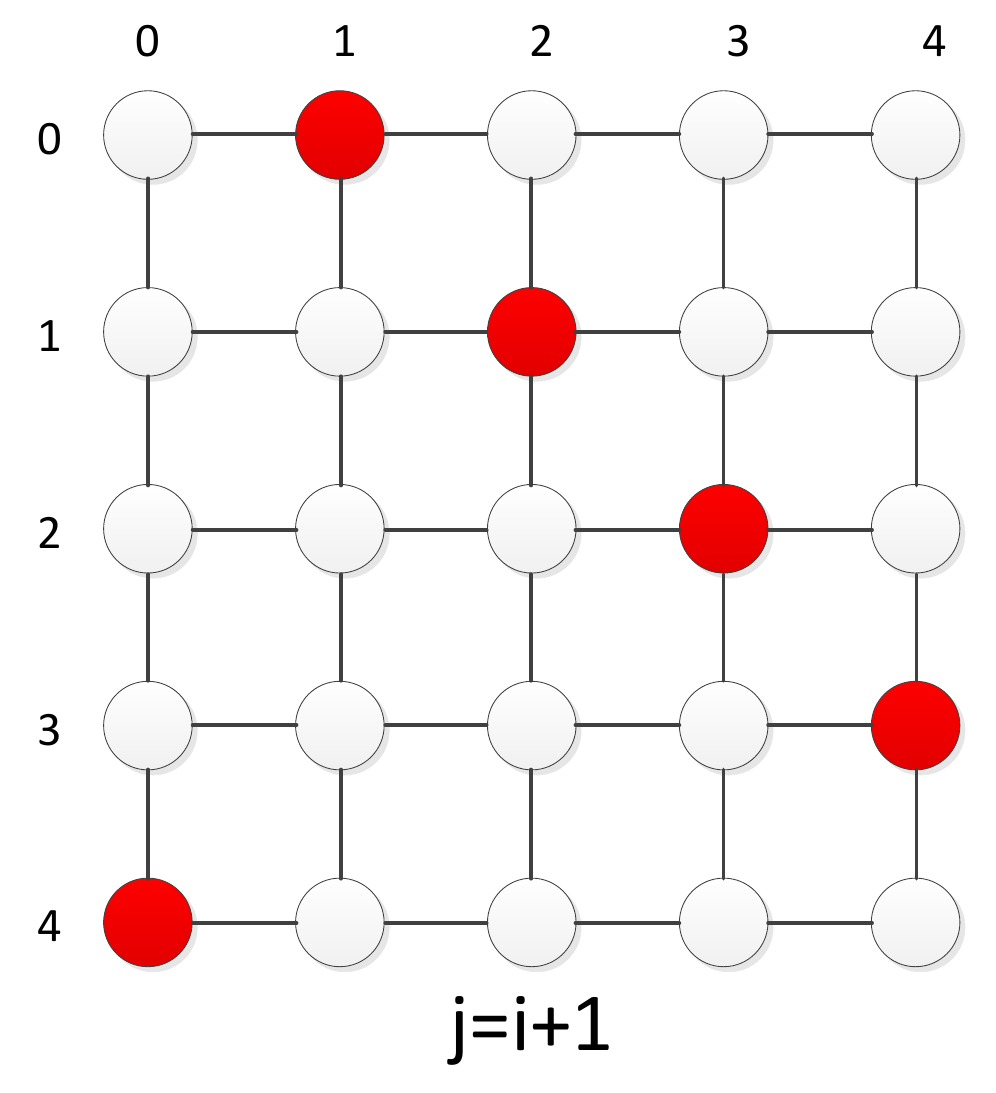}}
 \subfigure[]{\label{}
 \includegraphics[width=3cm]{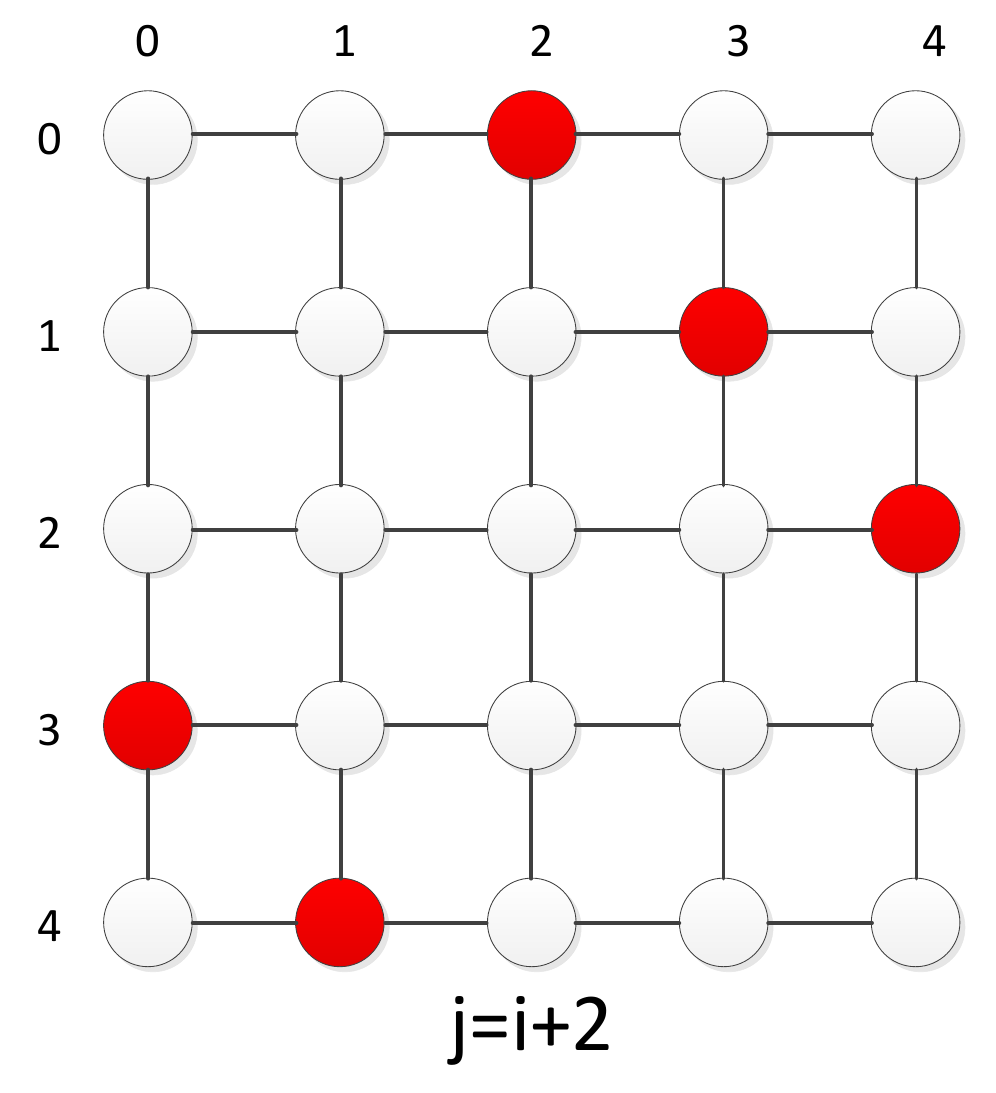}}
 \subfigure[]{\label{}
 \includegraphics[width=3cm]{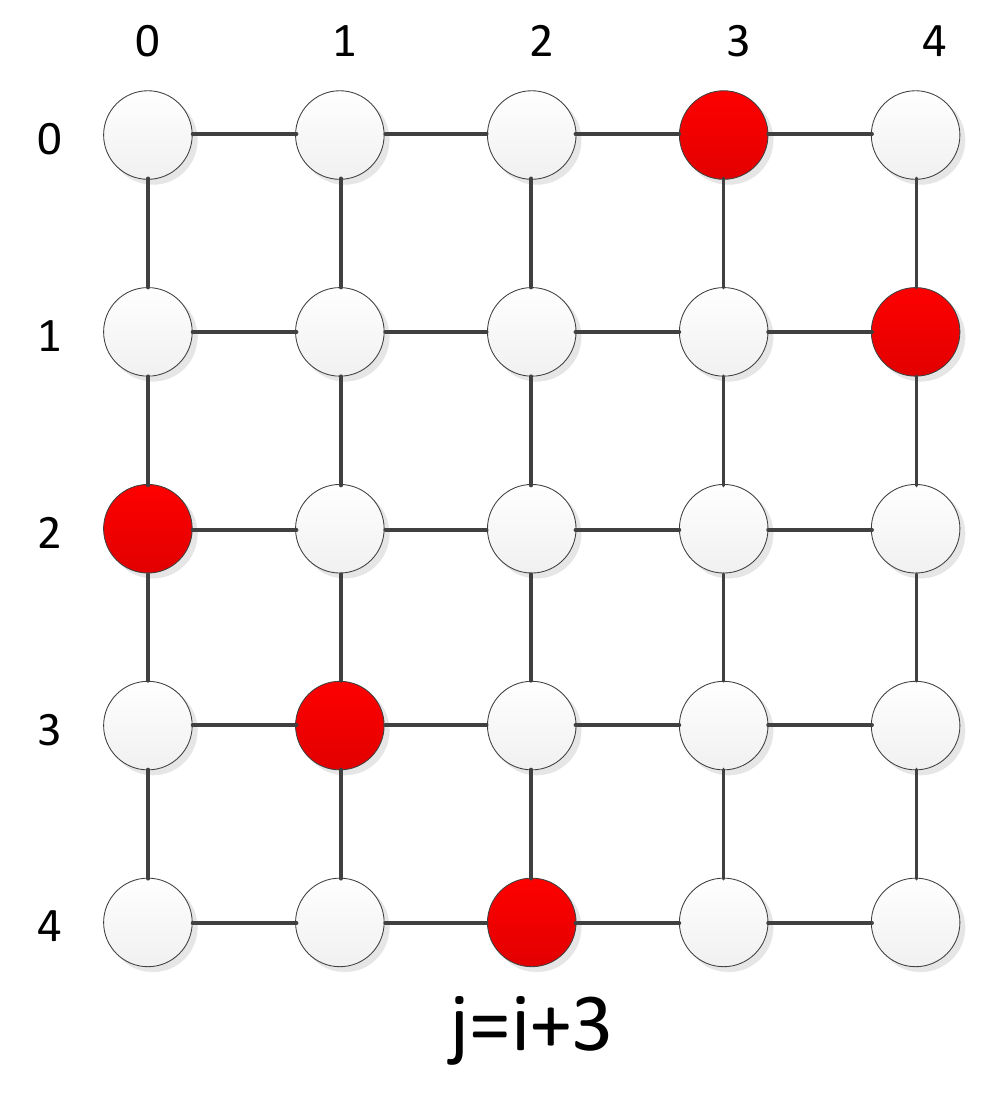}}
 \subfigure[]{\label{}
 \includegraphics[width=3cm]{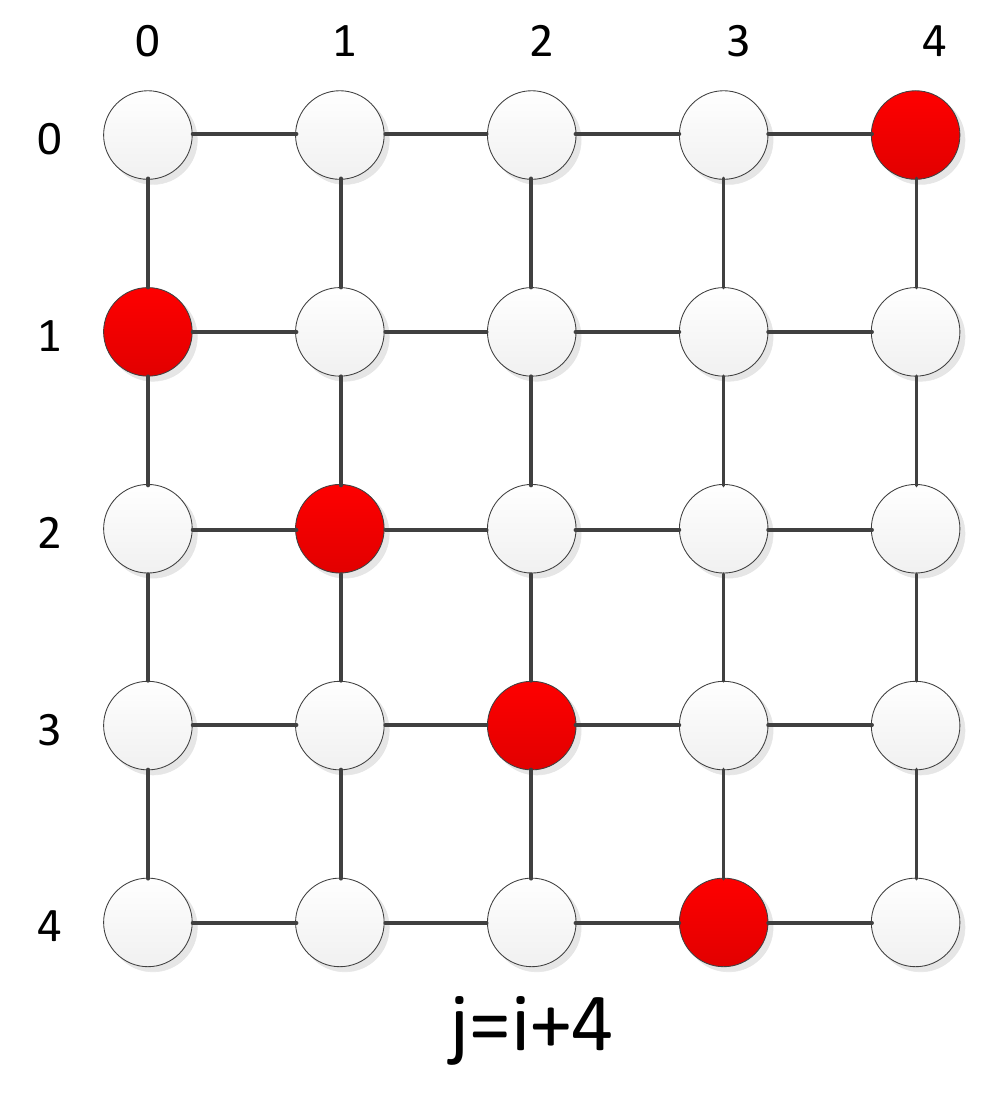}}
 \subfigure[]{\label{}
 \includegraphics[width=3cm]{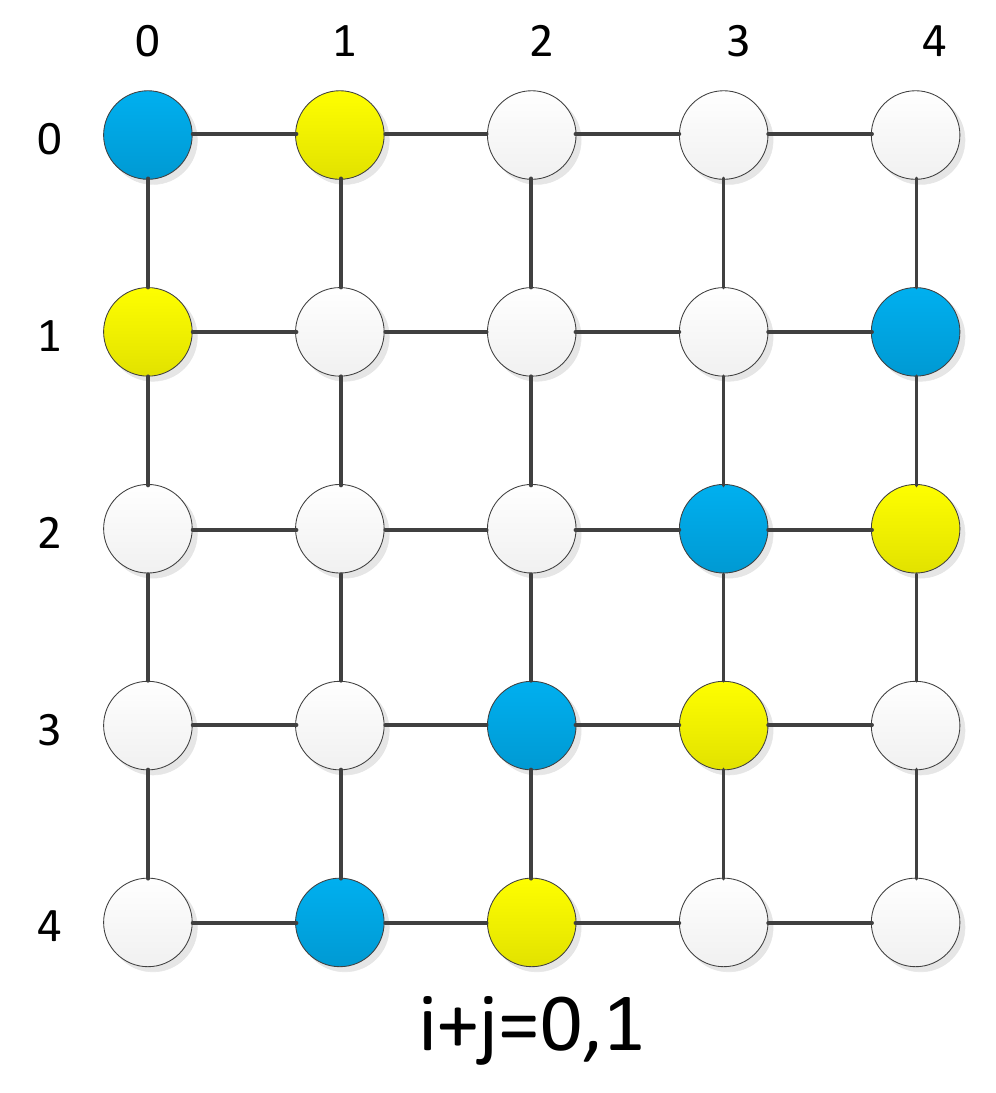}}
 \subfigure[]{\label{}
 \includegraphics[width=3cm]{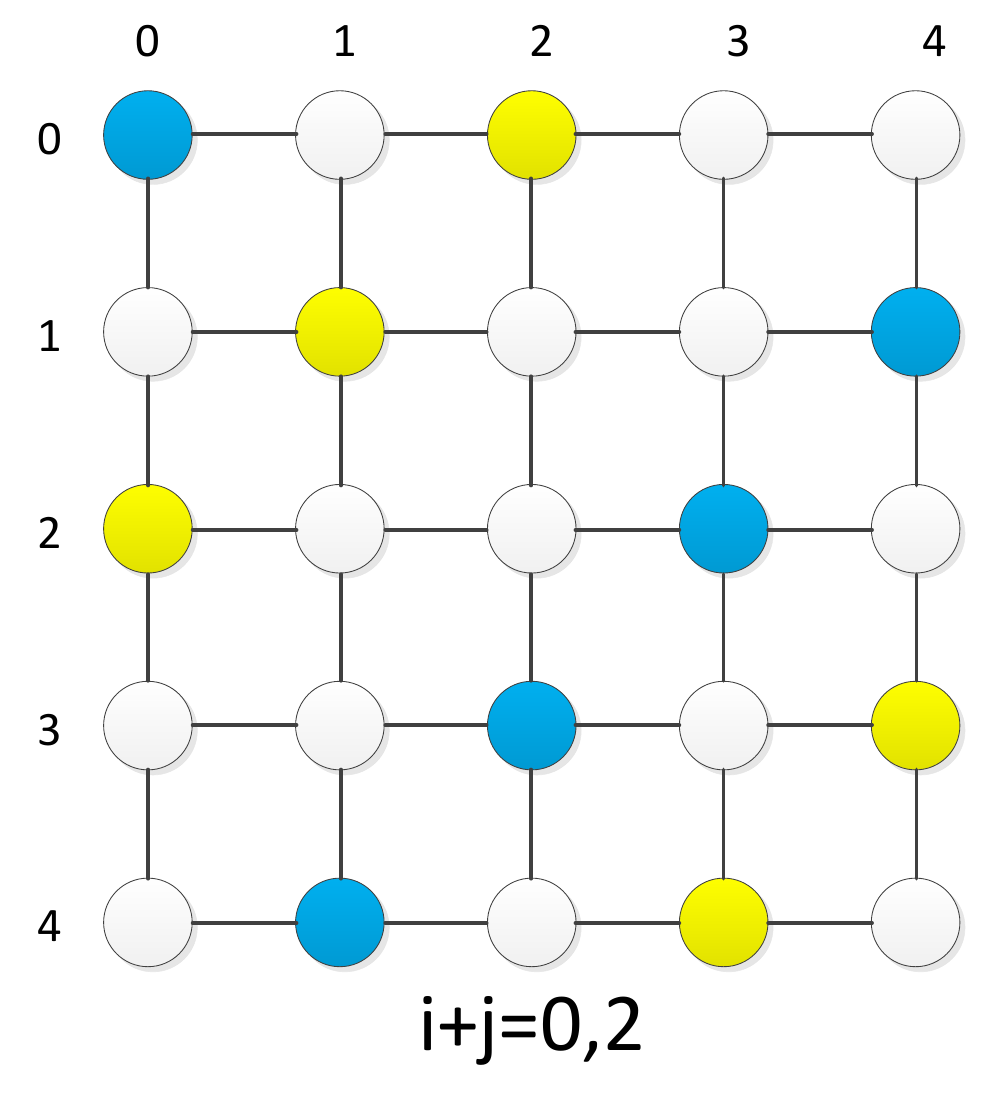}}
 \subfigure[]{\label{}
 \includegraphics[width=3cm]{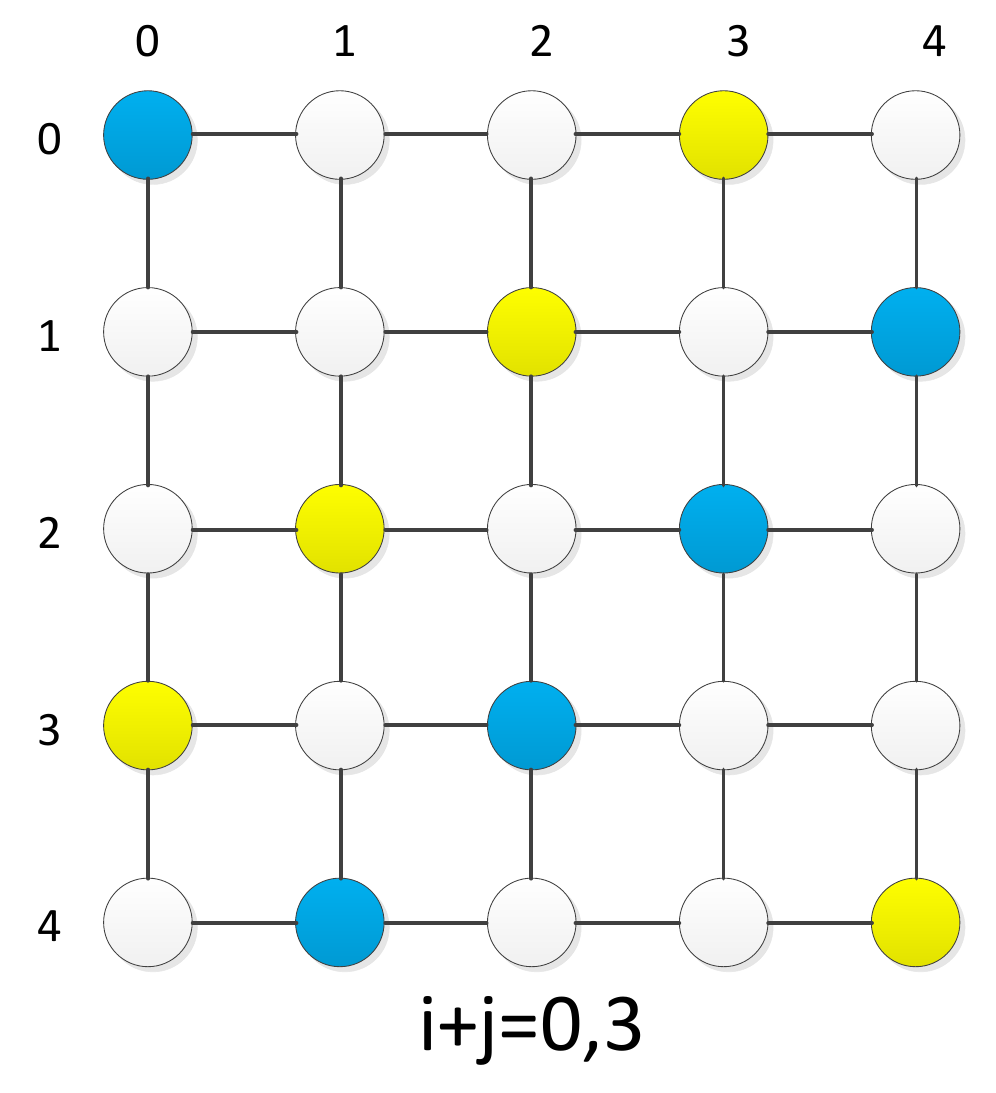}}
 \subfigure[]{\label{}
 \includegraphics[width=3cm]{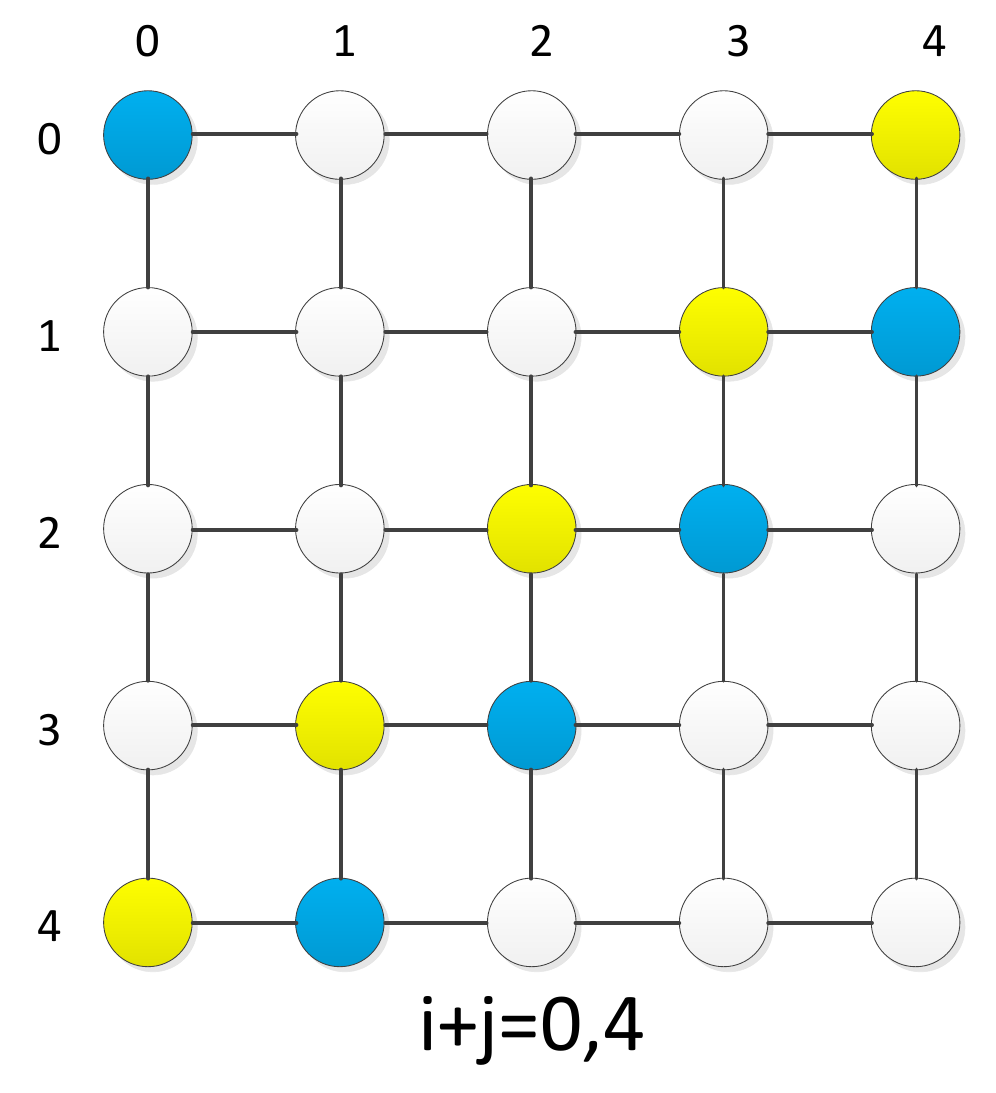}}
 \subfigure[]{\label{}
 \includegraphics[width=3cm]{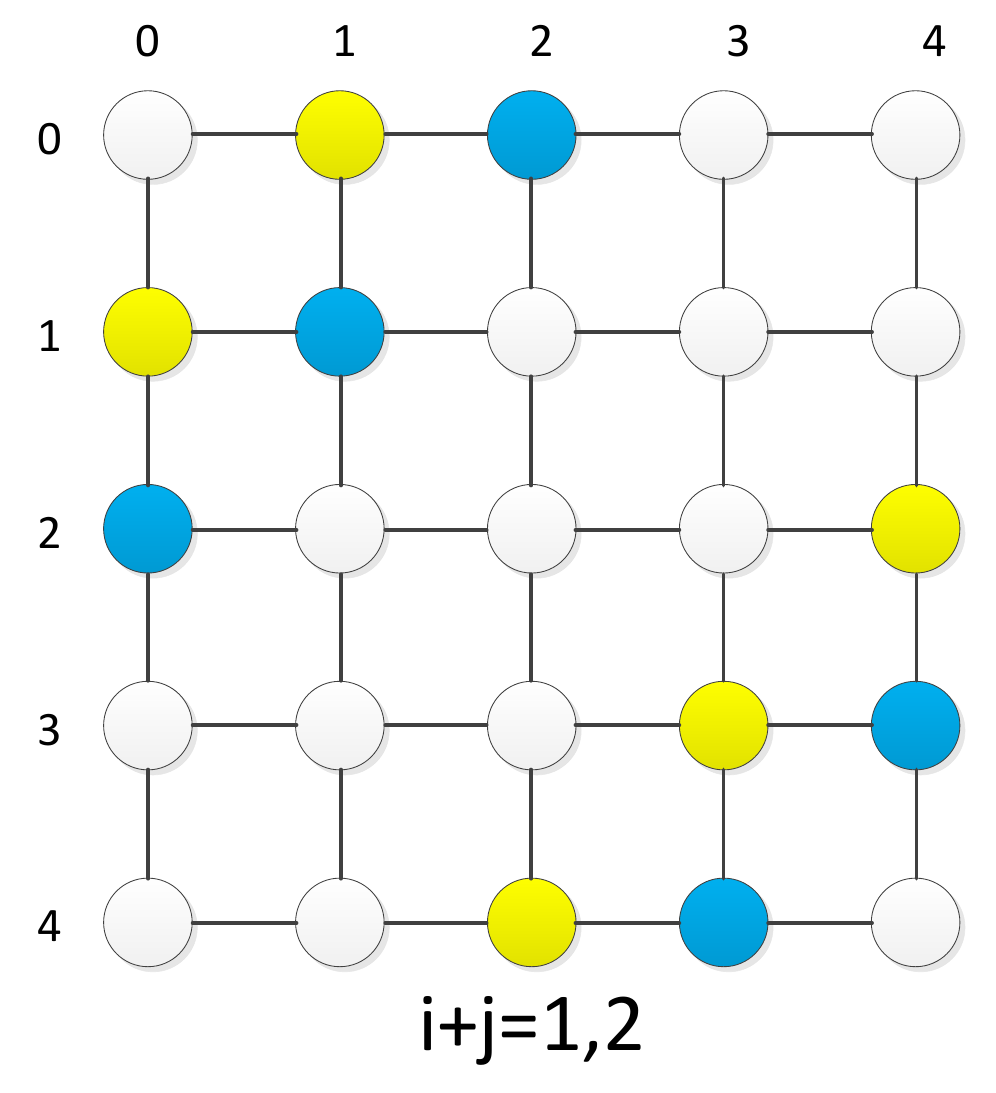}}
 \caption{The exceptional configurations on $5\times 5$ grid. (a)-(e) represent the cases when $j-i=\alpha$. (f)-(j) represent the cases when $j+i=\alpha, \beta$.}
 \label{55grid}
\end{figure}


\begin{theorem}
For an arbitrary non-negative integer $t$, if we run the quantum walks on the $N\times N$ grid for $t$ steps, where the marked vertices are $(i,j)$, $j-i=\alpha$, $i, j=0,1,\ldots,N-1$ and $\alpha$ is an arbitrary integer, the probability of finding a marked vertex is $\frac{1}{N}$.
\end{theorem}\label{extensiontheorem}

\begin{proof}
To begin with, let
\begin{equation}
|\psi_{\leftarrow,j}\rangle= \frac{1}{\sqrt{2N}}\Sigma_{i}(|\leftarrow,i,i+j  \rangle +|\downarrow,i,i+j \rangle),
\end{equation}
\begin{equation}
|\psi_{\rightarrow,j}\rangle=\frac{1}{\sqrt{2N}}\Sigma_{i}(|\rightarrow,i,i+j \rangle +|\uparrow,i,i+j \rangle).
\end{equation}
Then, the initial state of the quantum walk can be written as
\begin{equation}
|\psi(0)\rangle=\frac{1}{\sqrt{2N}}\Sigma_{j}(|\psi_{\leftarrow,j}\rangle +|\psi_{\rightarrow,j}\rangle).
\end{equation}
If $\alpha=0$, it is the known diagonal configuration~\cite{ambainis2008exceptional}.
Here, we try to extend $\alpha$ from 0 to any integer.
According to the AKR algorithm ($(G,-I)$-Type coined quantum walk) and mathematical induction, we get
\begin{equation}
|\psi(t)\rangle=\frac{1}{\sqrt{2N}}(\sum_{j=\alpha}^{N-t-1+\alpha}|\psi_{\leftarrow,j} \rangle -\sum_{j=N-t+\alpha}^{N+\alpha-1}|\psi_{\leftarrow,j} \rangle -\sum_{j=\alpha+1}^{t+\alpha}|\psi_{\rightarrow,j} \rangle  +\sum_{j=t+\alpha+1}^{N+\alpha}|\psi_{\rightarrow,j} \rangle),
\end{equation}
where $t<N$ and the operation is module N.
If we run the walks for $M=xN+y$ steps where $(y<N)$, we have the final state $|\psi(M)\rangle=(-1)^{x}|\psi(y)\rangle.$
In conclusion, the system only evolves by flipping signs and stays in a uniform probability distribution all the time.
Therefore, the probability of finding a marked vertex $(j=\alpha)$ is $\frac{1}{N}.$
\end{proof}

\begin{remark}
For an arbitrary integer $\alpha$, if the $N$ marked vertices $(i,j)$ satisfy the condition $i+j=\alpha$, the success probability is $\frac{1}{N}$ over time.
\end{remark}

\begin{proof}
Let
\begin{equation} |\phi_{\leftarrow,j}\rangle=\frac{1}{\sqrt{2N}}\sum_{i}(|\leftarrow,i,j-i\rangle+|\downarrow,i,j-i\rangle) \end{equation}
\begin{equation}
|\phi_{\rightarrow,j}\rangle=\frac{1}{\sqrt{2N}}\sum_{i}(|\rightarrow,i,j-i\rangle+|\uparrow,i,j-i\rangle)
\end{equation}
and then the initial state of the quantum walk is:
\begin{equation}
|\phi(0)\rangle=\frac{1}{\sqrt{4N^{2}}}\sum_{i,j,d}|d,i,j\rangle=\frac{1}{\sqrt{2N}} \sum_{j} (|\phi_{\leftarrow,j}\rangle+|\phi_{\rightarrow,j}\rangle).
\end{equation}
Using the similar proof above, it is easy to obtain the final conclusion.
\end{proof}

In fact, the above images (b)-(e) can also be seen as the translation and the rotation invariance of image (a) by the AKR algorithm on two-dimensional grid.
In addition, we find a combination of $j-i=\beta$ and $j-i=\alpha$, where $(i,j)$ is the coordinate of the marked vertex, would result in a new exceptional configuration for AKR algorithm.
Of course, the conclusion still holds if $j+i=\beta$ and $j+i=\alpha$.
These cases are shown in Figure \ref{55grid} (f)-(j).
Thus, we get a new class of exceptional configuration, which extends the range of exceptional configuration known before.
We elaborate on this in the following theorem.

\begin{theorem}
Suppose that the coordinate of the $2N$ marked vertices $(i,j)$ satisfies $j-i=\beta$ and $j-i=\alpha$, where $k=\beta-\alpha>0$,
then it is an exceptional configuration when $k=\lfloor\frac{N}{2}\rfloor$, where $\lfloor\cdot\rfloor$ is the function that its value is the largest integer less than or equal to the independent variable.
\end{theorem}

\begin{proof}
Here, we continue to use the notation presented in the Theorem 1. According to the model of quantum walk search on the grid ($(G,-I)$-Type quantum walk), the initial state
$|\psi(0)\rangle=\frac{1}{\sqrt{2N}}\Sigma_{j}(|\psi_{\leftarrow,j}\rangle +|\psi_{\rightarrow,j}\rangle)$ evolves by flipping signs.
Now let us discuss the odd and even case of $N$ in details based on the rule of AKR algorithm ($(G,-I)$-Type coined quantum walk) and mathematical induction.

When $N$ is an even number, i.e. $N=2\lfloor\frac{N}{2}\rfloor=2k$, the evolutionary process can be expressed as:
\begin{equation}\label{lsteps}
  \begin{split}
|\psi(l)\rangle=\frac{1}{\sqrt{2N}}&(-\sum_{j=\alpha+1}^{\alpha+l}|\psi_{\rightarrow,j}\rangle
                                     -\sum_{j=\beta+1}^{\beta+l}|\psi_{\rightarrow,j}\rangle
                                     +\sum_{j=\alpha+l+1}^{\beta}|\psi_{\rightarrow,j}\rangle
                                     +\sum_{j=\beta+l+1}^{N+\alpha}|\psi_{\rightarrow,j}\rangle\\
                                   &-\sum_{j=\alpha-l}^{\alpha-1}|\psi_{\leftarrow,j}\rangle
                                    -\sum_{j=\beta-l}^{\beta-1}|\psi_{\leftarrow,j}\rangle
                                    +\sum_{j=\alpha}^{\beta-l-1}|\psi_{\leftarrow,j}\rangle
                                    +\sum_{j=\beta}^{N+\alpha-l-1}|\psi_{\leftarrow,j}\rangle),
  \end{split}
\end{equation}
\begin{equation}
|\psi(k)\rangle=-|\psi(0)\rangle,
\end{equation}
where $0<l<k$.
So the period of this whole evolution is $2k$.
And the system only evolves by flipping signs.

When $N$ is an odd number, i.e. $N=2\lfloor\frac{N}{2}\rfloor+1=2k+1$,
the evolution of the first $k-1$ steps is the same as the previous case shown in Eq.~(\ref{lsteps}).
And then the evolutionary process can be expressed as:
\begin{equation}
|\psi(k)\rangle=\frac{1}{\sqrt{2N}}(|\psi_{\rightarrow,\alpha}\rangle -\sum_{j\neq\alpha}|\psi_{\rightarrow,j}\rangle+|\psi_{\leftarrow,\beta}\rangle-\sum_{j\neq\beta}|\psi_{\leftarrow,j}\rangle),
\end{equation}
\begin{equation}
|\psi(k+l')\rangle=\frac{1}{\sqrt{2N}}(|\psi_{\rightarrow,\alpha+l'}\rangle -\sum_{j\neq\alpha+l'}|\psi_{\rightarrow,j}\rangle+|\psi_{\leftarrow,\beta-l'}\rangle-\sum_{j\neq\beta-l'}|\psi_{\leftarrow,j}\rangle),
\end{equation}
\begin{equation}
|\psi(2k)\rangle=\frac{1}{\sqrt{2N}}(|\psi_{\rightarrow,\beta}\rangle -\sum_{j\neq\beta}|\psi_{\rightarrow,j}\rangle+|\psi_{\leftarrow,\alpha}\rangle-\sum_{j\neq\alpha}|\psi_{\leftarrow,j}\rangle),
\end{equation}
\begin{equation}
  \begin{split}
|\psi(2k+l'')\rangle=\frac{1}{\sqrt{2N}}&(\sum_{j=\beta+1}^{\beta+l''}|\psi_{\rightarrow,j}\rangle
                                     +\sum_{j=\alpha+1}^{\alpha+l''-1}|\psi_{\rightarrow,j}\rangle
                                     -\sum_{j=\beta+l''+1}^{N+\alpha}|\psi_{\rightarrow,j}\rangle
                                     -\sum_{j=N+\alpha+l''}^{N+\beta}|\psi_{\rightarrow,j}\rangle\\
                                   &+\sum_{j=\alpha-l''}^{\alpha-1}|\psi_{\leftarrow,j}\rangle
                                    +\sum_{j=\beta-l''+1}^{\beta-1}|\psi_{\leftarrow,j}\rangle)
                                    -\sum_{j=\alpha}^{\beta-l''}|\psi_{\leftarrow,j}\rangle
                                    -\sum_{j=\beta}^{N+\alpha-l''-1}|\psi_{\leftarrow,j}\rangle,
  \end{split}
\end{equation}
\begin{equation}
|\psi(3k+1)\rangle=|\psi(0)\rangle,
\end{equation}
where $0<l'<k$ and $0<l''\leq k$.
Thus the period of this whole evolution is $3k+1$.
It is easy to see that the quantum state at any given time is just a sign change compared with the initial state.

To sum up, the whole quantum system only evolves by flipping signs.
Also the success probability is always $\frac{2}{N}$, which is the same as their probability in the initial state over time.
Therefore, the arrangement of marked vertices satisfying $j-i=\beta$ and $j-i=\alpha$ and $\beta-\alpha=\lfloor\frac{N}{2}\rfloor$, is an exceptional configuration.

\end{proof}

\subsection{Two modified quantum walks}
As stated above, there are many exceptional configurations on the $N\times N$ grid under $(G,-I)$-Type coined quantum walk. In such situations, the success probability of the AKR algorithm will not grow over time.

In this section, we develop two promising approaches, i.e. two kinds of modified quantum walks that can find the marked vertices in the above configurations with a higher probability in terms of numerical simulation.
Through the numerical simulation results, we can find the approximate running time (steps) to obtain a higher success probability than the exceptional case.

For convenience, we refer to the AKR algorithm~\cite{ambainis2005coins} as algorithm $\mathcal{A}$. In the following subsections, we will propose two algorithms based on quantum walks which are called algorithm $\mathcal{B}$ and algorithm $\mathcal{C}$ respectively. The coin flip transformations of the above three algorithms are listed in Table~\ref{Three algorithms}. Here we mainly consider the exceptional configuration, where the coordinate of the marked vertex $(i,j)$ satisfies $j-i=\alpha$, and $\alpha$ is an arbitrary integer. If $i+j=\alpha$, it can improve in the similar way. And thus we will not repeat.

\begin{table}[!t]
\footnotesize
\caption{The coin flip transformations of three algorithms.}
\label{Three algorithms}
\tabcolsep 26pt 
\begin{tabular*}{\textwidth}{ccccc}
\toprule
  $Vertex \backslash Algorithm$ & $\mathcal{A}$ & $\mathcal{B}$ & $\mathcal{C}$ (odd step)  & $\mathcal{C}$ (even step) \\\hline
  unmarked vertices              & G  & G  & G            & G     \\
  marked vertices                & -I & F  & -I           & F      \\
\bottomrule
\end{tabular*}
\end{table}

\subsubsection{Quantum walk with one coin}
In this subsection, we develop a new quantum walk mechanism. Compared with algorithm $\mathcal{A}$, we use $F$ as the coin flip transformation instead of $-I$, where $F$ denotes the Fourier transformation,
\begin{equation}
F=\frac{1}{2}
\left(
\begin{array}{cccc}
 1  &  1   &  1   &  1 \\
 1  &  i   &  -1  &  -i \\
 1  &  -1  &  1   &  -1 \\
 1  &  -i  &  -1  &  i
\end{array}
\right).
\end{equation}
Thus, algorithm $\mathcal{B}$ based on quantum walk with one coin is implemented as follows,
\begin{itemize}
\itemindent 2.8em
\item [(1)] Initialize the walker to the equal superposition state given by Eq.~(\ref{gridinitial}).
\item [(2)] For the coin oracle $C'$, the coin flip operators $C'_{0}=G$ and $C'_{1}=F$ are applied to the unmarked and marked vertices respectively. The unitary operator $U'=S\cdot C'$ is performed for $t$ times.
\item [(3)] Measure the state of walker in the $|d,i,j\rangle$ basis, and obtain the success probability by calculating inner product between the final state and the marked state.
\end{itemize}

Based on the above settings, the unitary operator in algorithm $\mathcal{B}$ can be described as:
\begin{equation}
\begin{split}
U'&=S\cdot C'\\
  &=\sum_{i,j}(|\downarrow,i,j-1\rangle\langle\uparrow,i,j|+|\uparrow,i,j+1\rangle\langle\downarrow,i,j|
   +|\rightarrow,i-1,j\rangle\langle\leftarrow,i,j|+|\leftarrow,i+1,j\rangle\langle\rightarrow,i,j|)\\
  &\ \ \ \ \cdot[G\otimes(I-\sum_{v}|v\rangle\langle v|)+F\otimes\sum_{v}|v\rangle\langle v|]\\
  &=\frac{1}{2}
    \left(
    \begin{array}{cccc}
     B_{3}+B_{4}   &  -B_{3}+iB_{4}  &  B_{3}-B_{4}   &  B_{3}-iB_{4}  \\
     -B_{1}+B_{2}  &  B_{1}+B_{2}    &  B_{1}+B_{2}   &  B_{1}+B_{2}   \\
     B_{7}+B_{8}   &  B_{7}-iB_{8}   &  B_{7}-B_{8}   &  -B_{7}+iB_{8} \\
     B_{5}+B_{6}   &  B_{5}-B_{6}    &  -B_{5}+B_{6}  &  B_{5}-B_{6}
   \end{array}\right)
\end{split}
\end{equation}
where $v\in M$, $M$ is the marked vertex set and
$$B_{1}=\sum_{(i,j)\not\in M}|i,j-1\rangle\langle i,j|, \ \ \ \ \ \ B_{2}=\sum_{(i,j)\in M}|i,j-1\rangle\langle i,j|,$$
$$B_{3}=\sum_{(i,j)\not\in M}|i,j+1\rangle\langle i,j|, \ \ \ \ \ \ B_{4}=\sum_{(i,j)\in M}|i,j+1\rangle\langle i,j|,$$
$$B_{5}=\sum_{(i,j)\not\in M}|i-1,j\rangle\langle i,j|, \ \ \ \ \ \ B_{6}=\sum_{(i,j)\in M}|i-1,j\rangle\langle i,j|,$$
$$B_{7}=\sum_{(i,j)\not\in M}|i+1,j\rangle\langle i,j|, \ \ \ \ \ \ B_{8}=\sum_{(i,j)\in M}|i+1,j\rangle\langle i,j|.$$

Since the matrix is too complex, it is hard to provide the analytical results.
However, numerical simulation can provide us with direct conclusions and enlightenment.
In order to obtain a more intuitive comparison and a relatively clear proof, we choose to use the numerical simulation method to figure it out.

Next, we consider $N$ marked vertices placed as an exceptional configuration, in which the coordinate ($i,j$) of marked locations on the $N\times N$ grid satisfies $j-i=0$ and $j-i=0, 1$, and $N=3, 16, 100, 300$. We compare the success probability of algorithm $\mathcal{A}$ and algorithm $\mathcal{B}$.

\begin{figure}[htb]
 \centering
 \subfigure[]{\label{sub1}
 \includegraphics[width=4.1cm]{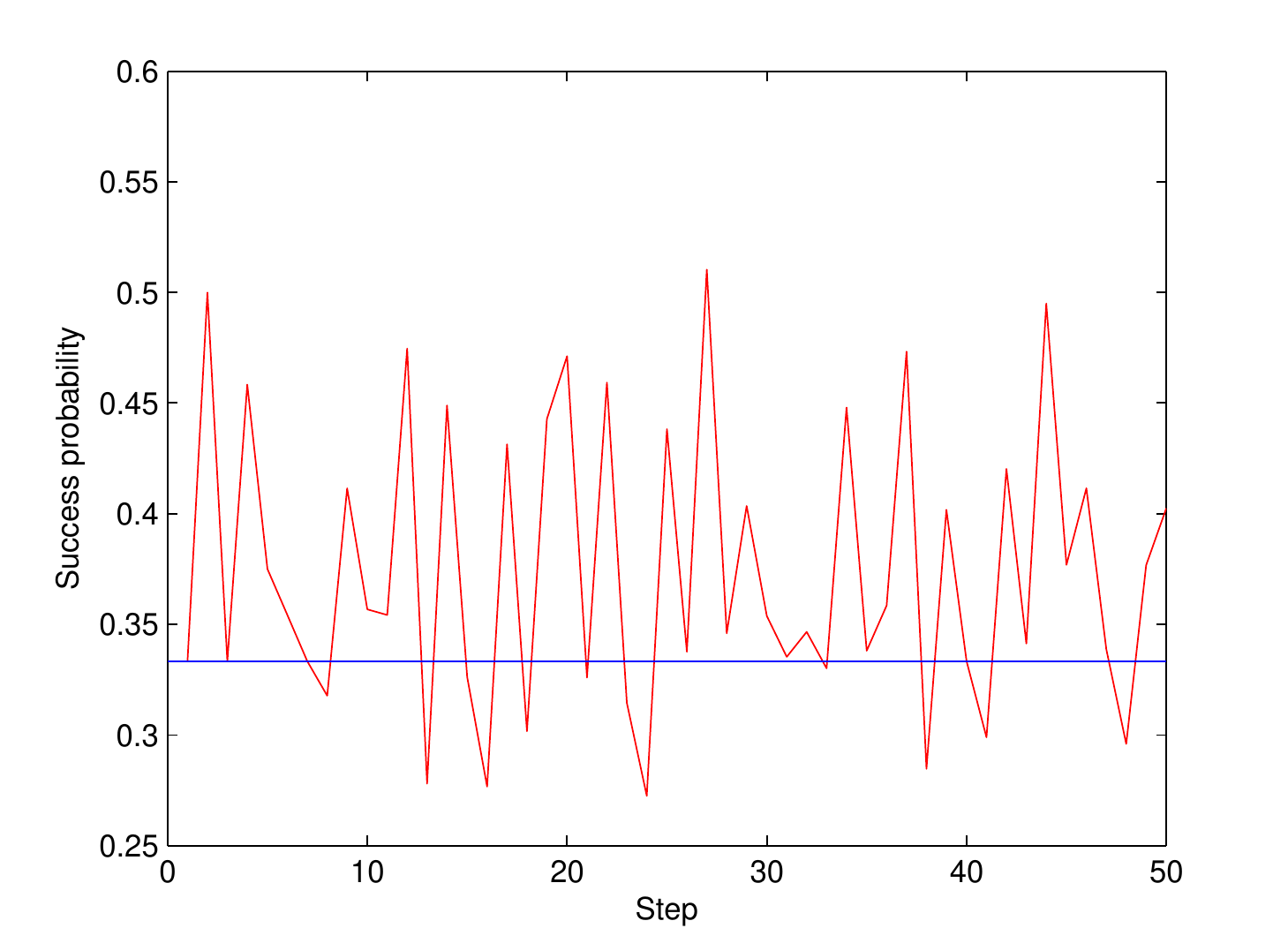}}
 \subfigure[]{\label{sub2}
 \includegraphics[width=4.1cm]{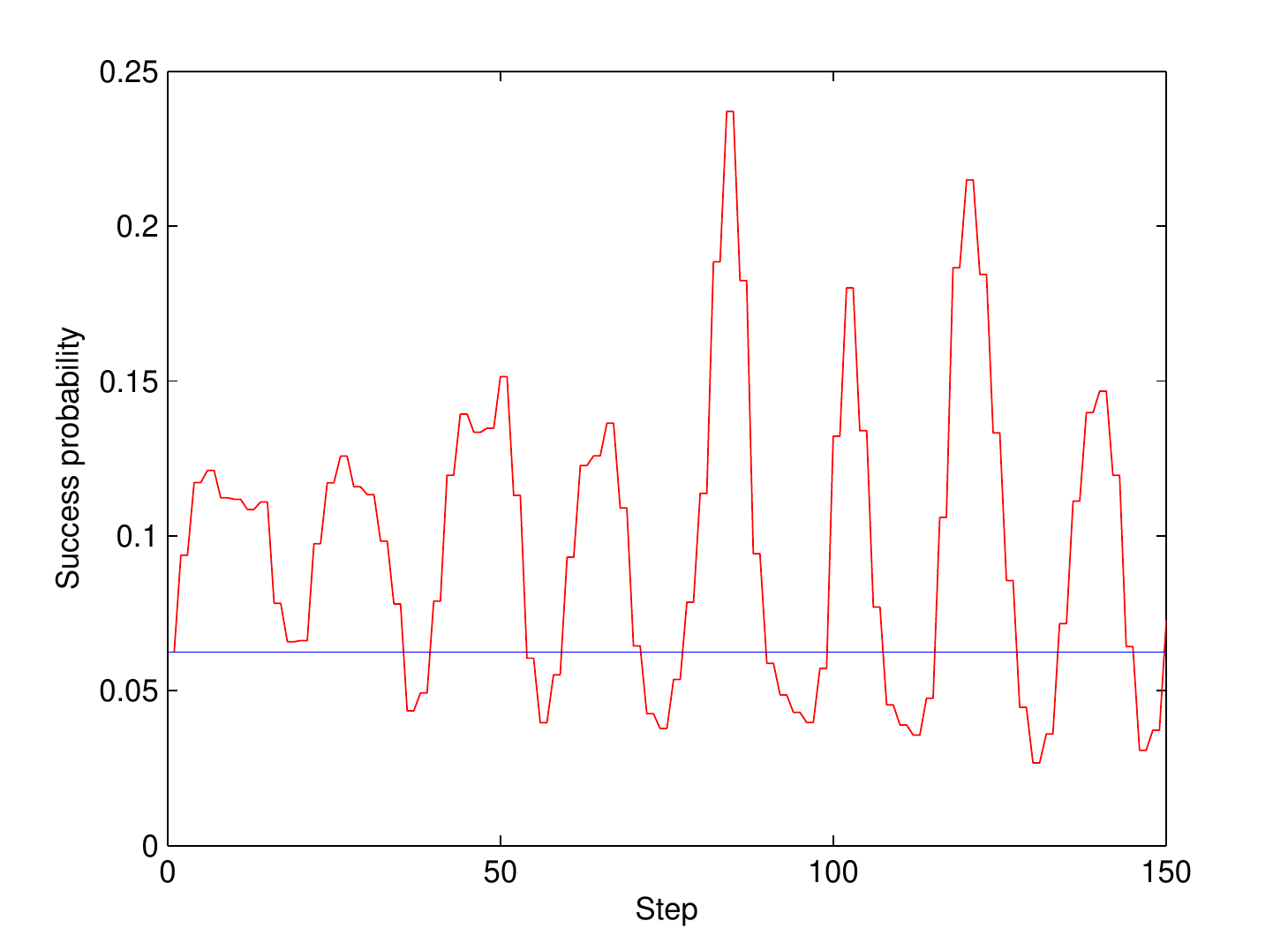}}
 \subfigure[]{\label{sub3}
 \includegraphics[width=4.1cm]{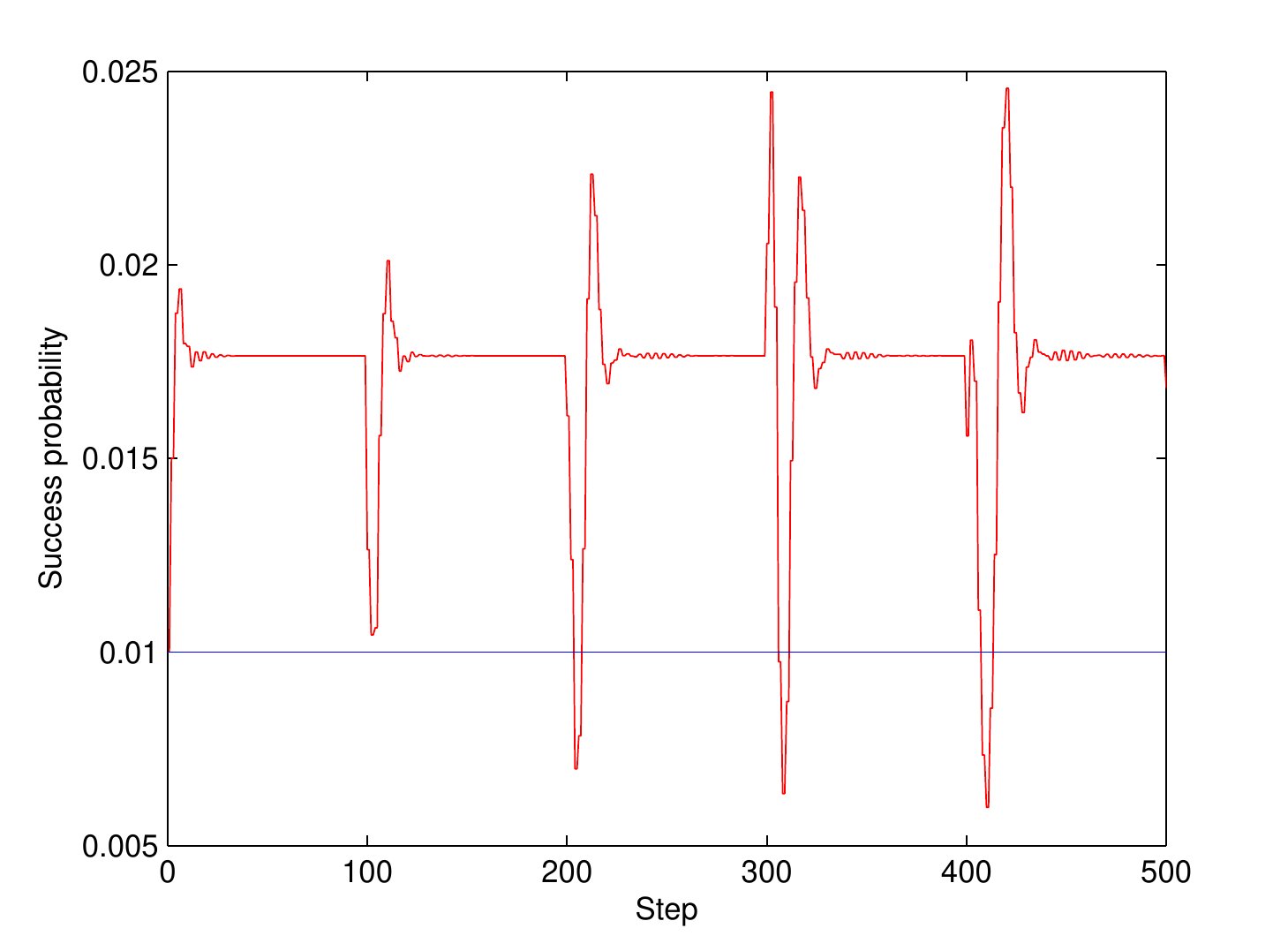}}
 \subfigure[]{\label{sub4}
 \includegraphics[width=4.1cm]{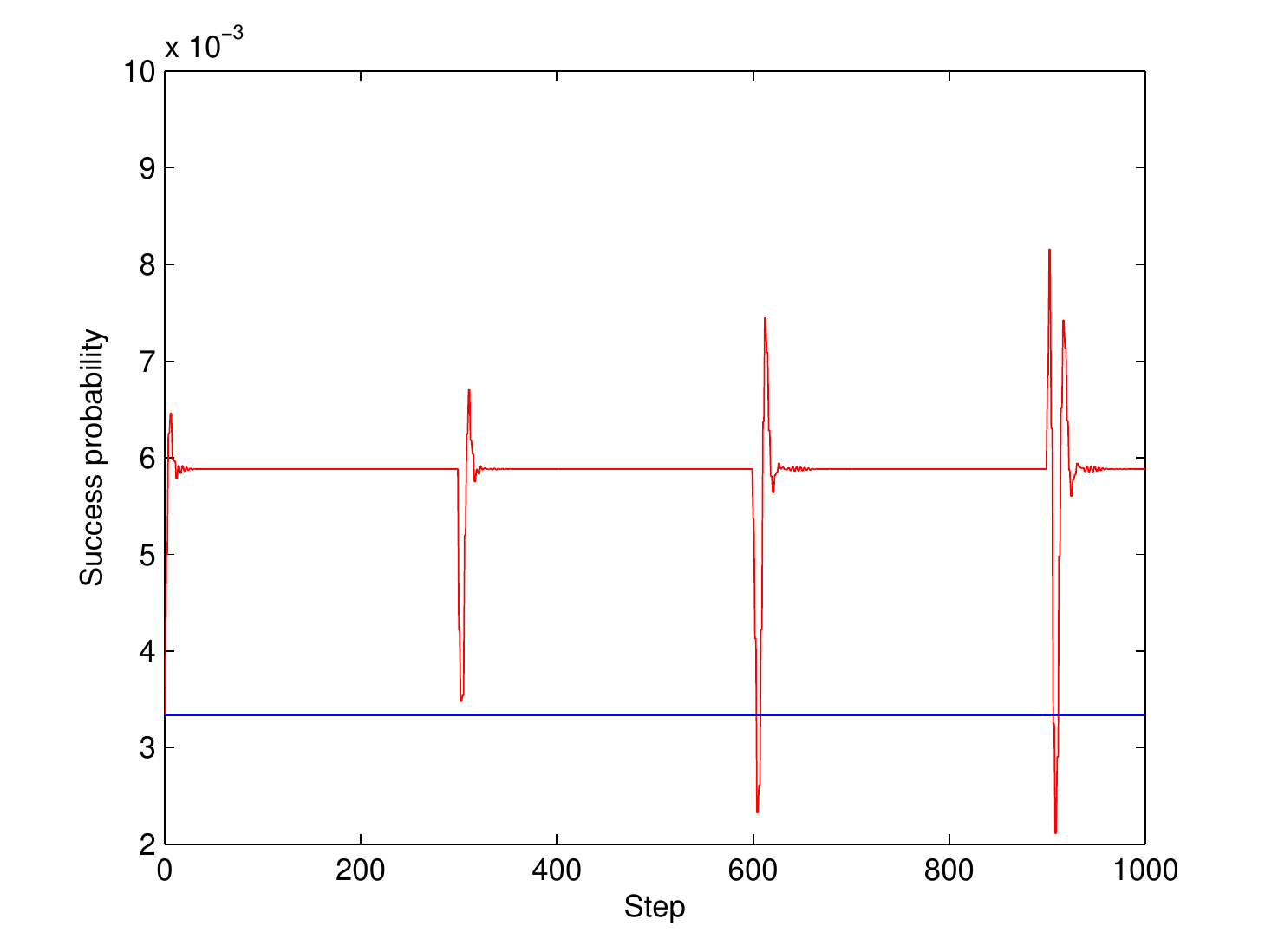}}
 \subfigure[]{\label{sub5}
 \includegraphics[width=4cm]{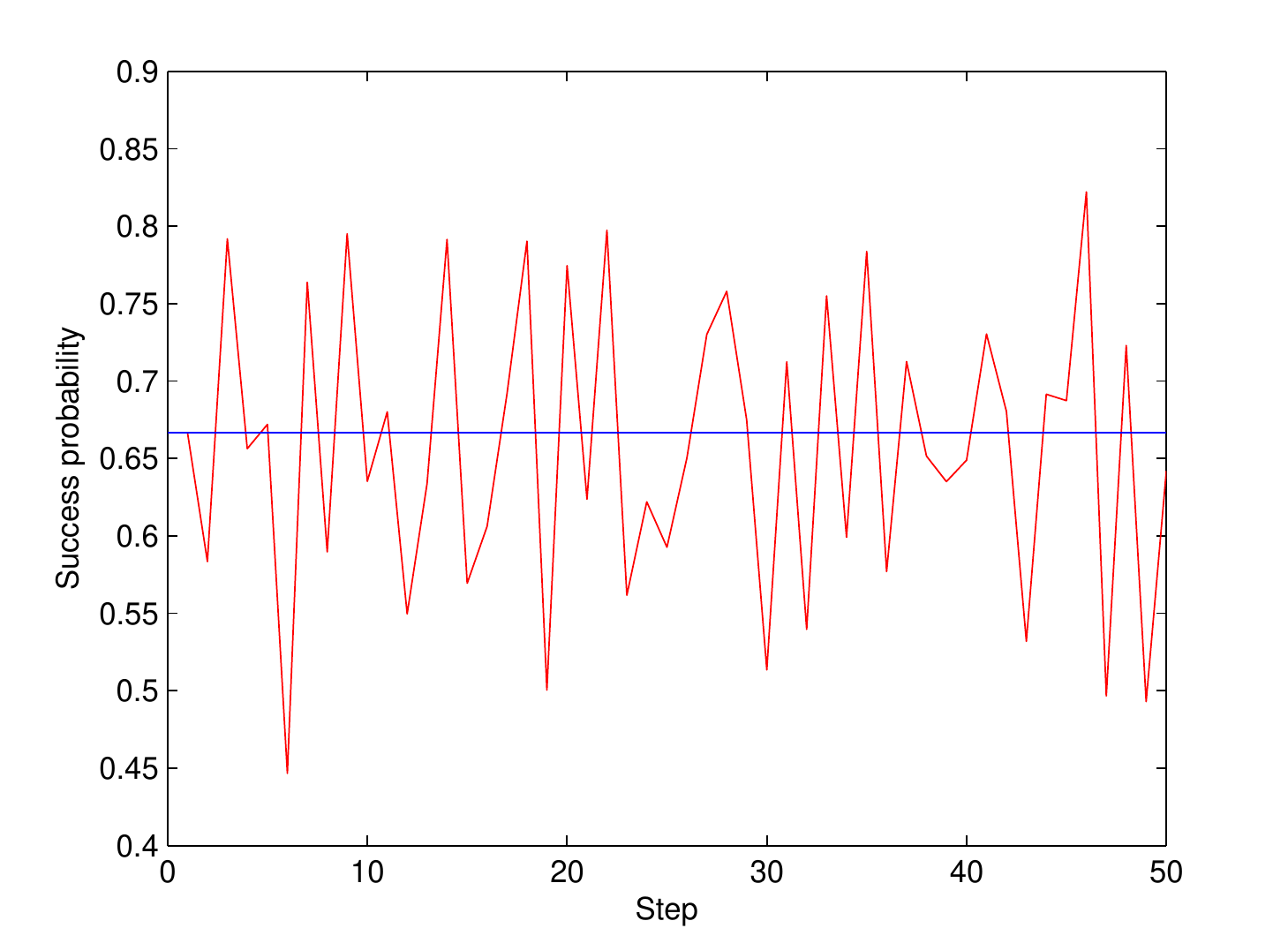}}
 \subfigure[]{\label{sub6}
 \includegraphics[width=4.1cm]{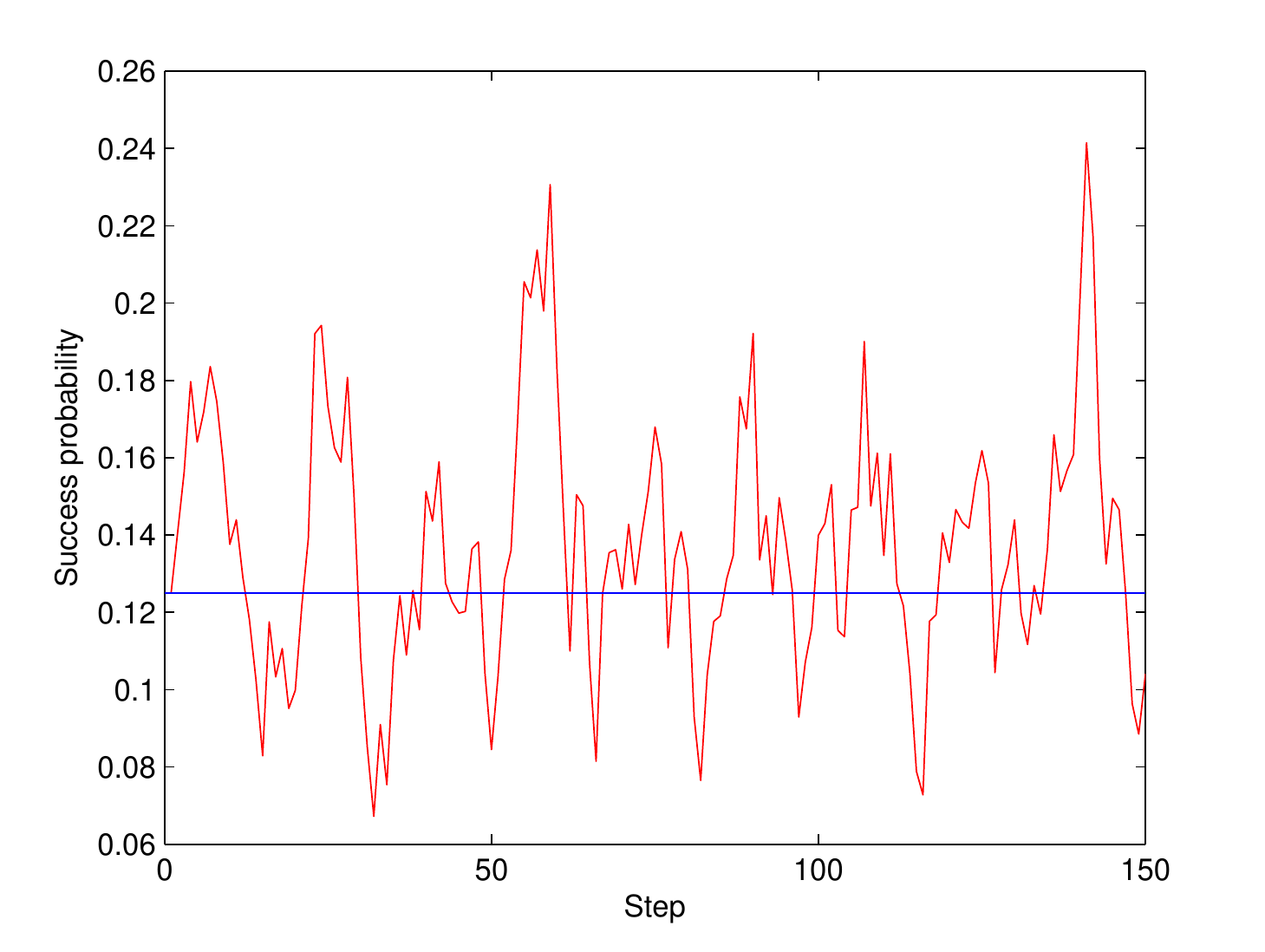}}
 \subfigure[]{\label{sub7}
 \includegraphics[width=4.1cm]{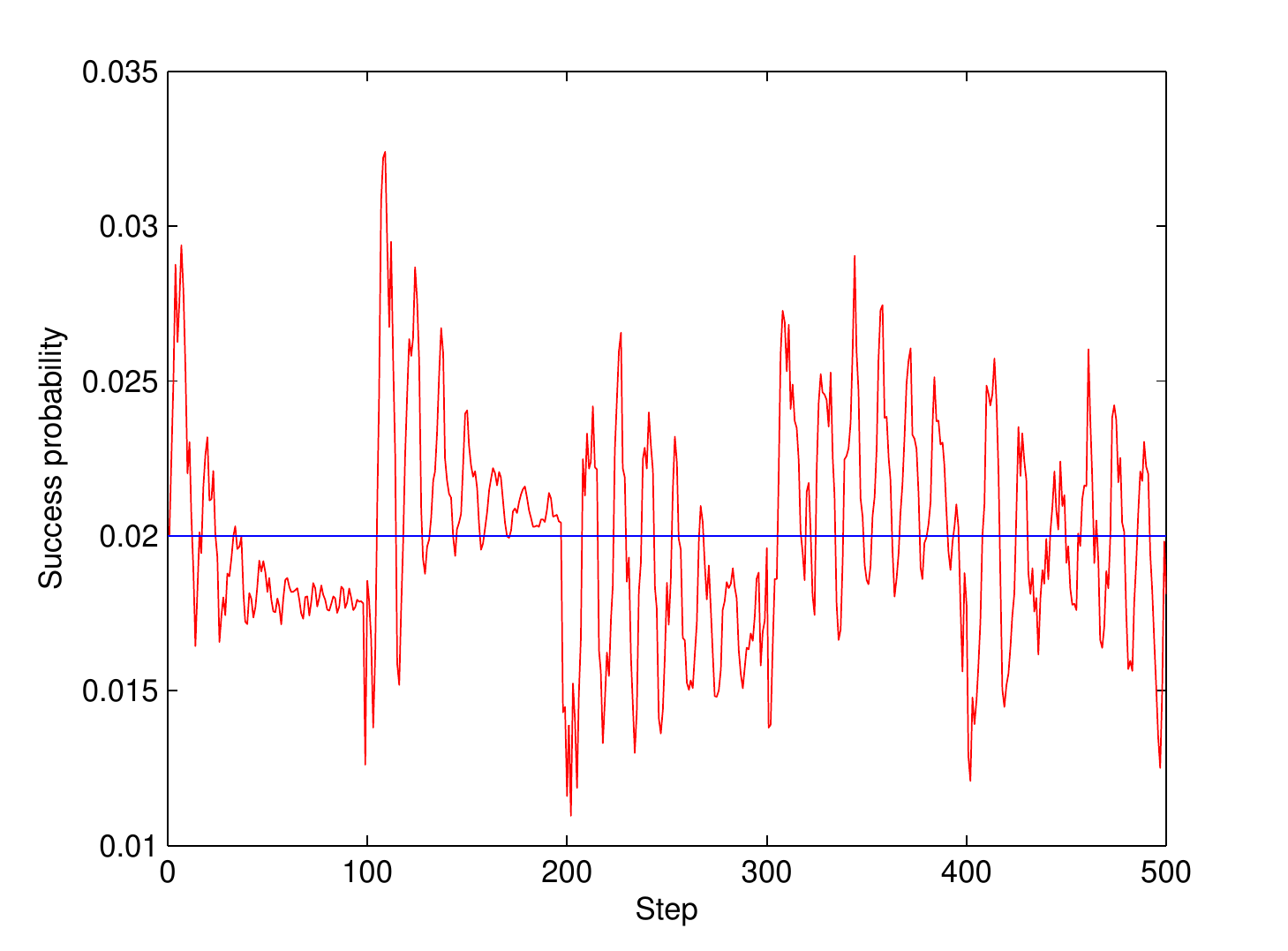}}
 \subfigure[]{\label{sub8}
 \includegraphics[width=4.1cm]{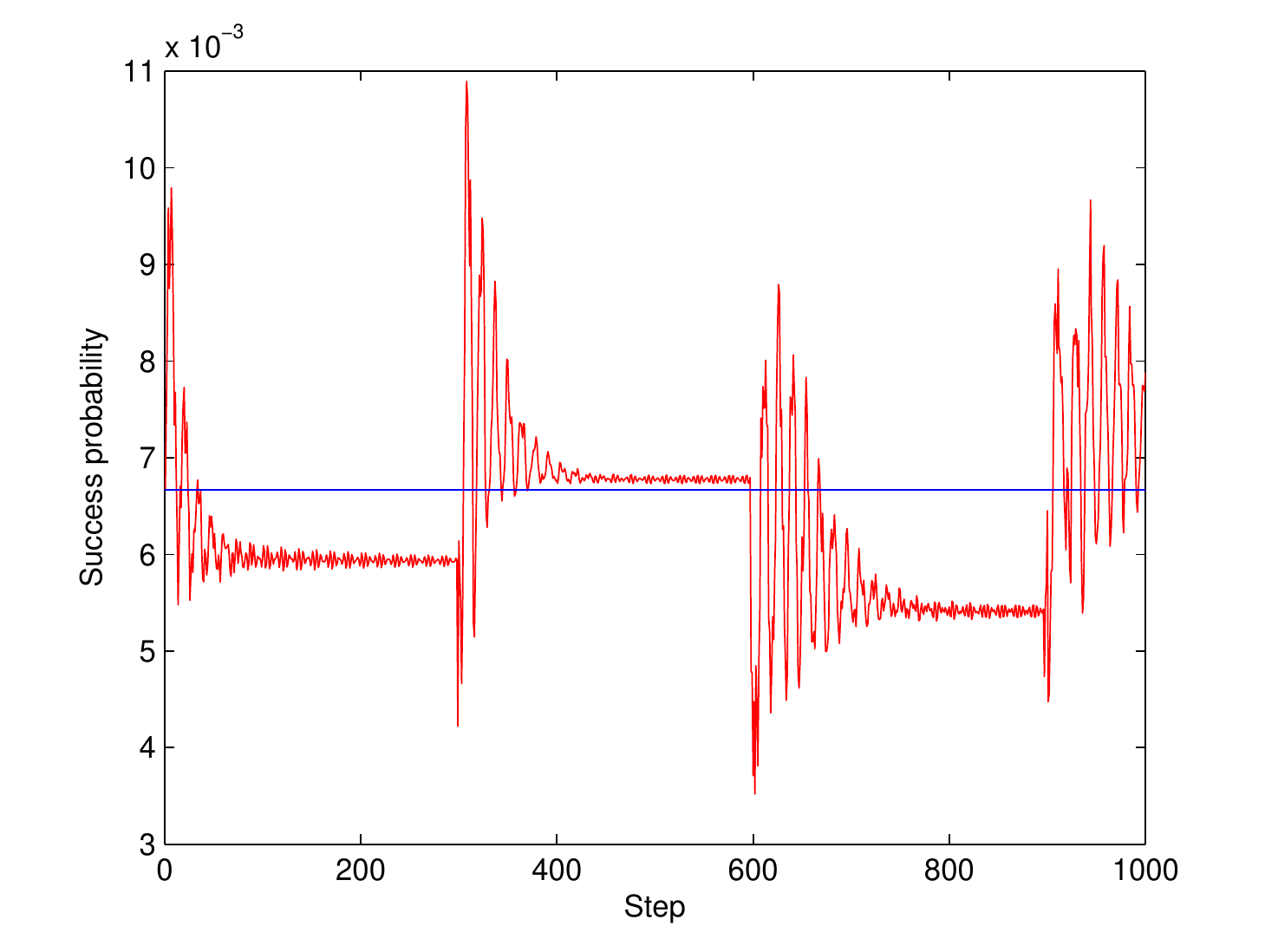}}
 \renewcommand{\figurename}{Figure}
 \caption{The success probability of finding the marked vertices on the two-dimensional grid. (a)(e): $3\times 3$ grid; (b)(f): $16\times 16$ grid; (c)(g): $100\times 100$ grid; (d)(h):$300\times 300$ grid. (a)-(d): $j-i=0$; (e)-(h): $j-i=0, 1$. The blue line corresponds to algorithm $\mathcal{A}$ and the red curve corresponds to algorithm $\mathcal{B}$.}
 \label{gfgrid}
\end{figure}

Figure~\ref{sub1} shows the results of numerical simulations for $3\times 3$ grid. We run the quantum walk 50 steps. From the figure, we know that the success probability of algorithm $\mathcal{A}$ is always 0.33, regardless of the number of the steps, which is consistent with the above discussion. However, the success probability of algorithm $\mathcal{B}$ is shocked, and it exceeds 0.5 after some steps. For example, the probability could reach 0.5103 after 27 steps.
Thus, as for the exceptional configuration, we could find the marked vertices with higher probability by this modified quantum walk.

In Figure~\ref{sub2}-\ref{sub4}, the results of numerical simulations for $16\times 16$, $100\times 100$ and $300\times 300$ grid are presented. We run the algorithms for 150 times, 500 times and 1000 times respectively. The blue line corresponds to algorithm $\mathcal{A}$, which shows the success probability does not change.
The success probability of algorithm $\mathcal{B}$ also reaches peak value which is higher than classically guessing at some special steps. Therefore, for this kind of the arrangement of marked vertices, algorithm $\mathcal{B}$ is better than algorithm $\mathcal{A}$.

For the case when the exceptional configuration is a combination of $j-i=0$ and $j-i=1$, we also discuss their performances under algorithm $\mathcal{B}$.
In particular, the numerical results on $3\times 3$ grid, $16\times 16$ grid, $100\times 100$ grid and $300\times 300$ grid for $50$, $150$, $500$, $1000$ times are shown in Figure \ref{sub5}-\ref{sub8}.
The success probability of algorithm $\mathcal{B}$ varies as the number of steps increases. Moreover, the success probability of most time moments is higher than algorithm $\mathcal{A}$, so we can choose appropriate time moment to take measurement.

\subsubsection{Quantum walk with two coins}
In this subsection, we present another model of quantum walk. For a quantum walk on the $N\times N$ grid driven by two coins, the movement of the walker is determined by the first coin $C^{*}$ and the second coin $C^{**}$ alternately.
The state of the quantum walk with two coins on the grid is given by
\begin{equation}
|\psi(t)\rangle=\Sigma_{i,j,d_{1},d_{2}}|d_{1},d_{2},i,j\rangle.
\end{equation}
The unitary transformations at odd step (flip the $1st$ coin) and even step (flip the $2nd$ coin) can be written as
\begin{equation}
  \begin{split}
  U_{1}&=S_{1}\cdot C^{*} =S_{1}\cdot (C^{*}_{0} \otimes I_{4}\otimes(I_{N^{2}}-\Sigma_{v} |v\rangle \langle v|) + C^{*}_{1} \otimes I_{4}\otimes\Sigma_{v} |v\rangle \langle v|), \\
  U_{2}&=S_{2}\cdot C^{**}=S_{2}\cdot (I_{4}\otimes C^{**}_{0}\otimes(I_{N^{2}}-\Sigma_{v} |v\rangle \langle v|)+ I_{4} \otimes C^{**}_{1}\otimes\Sigma_{v} |v\rangle \langle v|),
  \end{split}
\end{equation}
where $|v\rangle$ is the marked vertex, $C^{*}_{0}=G$, $C^{*}_{1}=-I$, $C^{**}_{0}=G$, $C^{**}_{1}=F$ in algorithm $\mathcal{C}$.
The shift transformations $S_{1}$ and $S_{2}$ are
\begin{equation}
\begin{array}{ccccccc}
 |\uparrow,d_{2},i,j \rangle     &  \xrightarrow{S_{1}}  &  |\downarrow,d_{2},i,j-1 \rangle  &\ \ \ & |d_{1},\uparrow,i,j \rangle     &  \xrightarrow{S_{2}}  &  |d_{1},\downarrow,i,j-1 \rangle   \\
 |\downarrow,d_{2},i,j \rangle   &  \xrightarrow{S_{1}}  &  |\uparrow,d_{2},i,j+1 \rangle  &\ \ \ &
 |d_{1},\downarrow,i,j \rangle   &  \xrightarrow{S_{2}}  &  |d_{1},\uparrow,i,j+1 \rangle  \\
 |\leftarrow,d_{2},i,j \rangle   &  \xrightarrow{S_{1}}  &  |\rightarrow,d_{2},i-1,j \rangle &\ \ \ &
  |d_{1},\leftarrow,i,j \rangle   &  \xrightarrow{S_{2}}  &  |d_{1},\rightarrow,i-1,j \rangle  \\
 |\rightarrow,d_{2},i,j \rangle  &  \xrightarrow{S_{1}}  &  |\leftarrow,d_{2},i+1,j \rangle  &\ \ \ &
  |d_{1},\rightarrow,i,j \rangle  &  \xrightarrow{S_{2}}  &  |d_{1},\leftarrow,i+1,j \rangle
\end{array}
\end{equation}
where $d_{1}, d_{2}\in \{ \uparrow, \downarrow, \leftarrow, \rightarrow \}$.
If we do a total $t$ flips, the corresponding operator is ($k=1,2,3,\ldots$)
\begin{equation}
  U=\left\{
  \begin{array}{ll}
     (U_{2}U_{1})^{k},          &\mbox{$t=2k$ ;}\\
      U_{1}(U_{2}U_{1})^{k-1},  &\mbox{$t=2k-1$.}
  \end{array}
  \right.
\end{equation}

Algorithm $\mathcal{C}$ based on quantum walk with two coins is implemented as follows,
\begin{itemize}
\item 1. Initialize the walker to the equal superposition over all states, given by
         \begin{equation}
         |\psi(0)\rangle=\frac{1}{\sqrt{16N^2}}\Sigma_{i,j,d_{1},d_{2}}|d_{1},d_{2},i,j\rangle.
         \end{equation}
\item 2. Apply the unitary operator $U_{1}=S_{1}\cdot C^{*}$ and $U_{2}=S_{2}\cdot C^{**}$ alternatively.
\item 3. Measure the state of the walker in the $|d_{1},d_{2},i,j\rangle$ basis, and obtain the probability of finding the marked vertices by calculating the inner product between the final state and the marked state.
\end{itemize}

Similar to algorithm $\mathcal{B}$, the specific evolution of algorithm $\mathcal{C}$ can also be written down in detail. Since it involves two coins, the unitary operator is a piecewise function and thus it is more complicated. I will not repeat here.
 And thus the numerical simulation may be required.
Next, we consider $N$ marked vertices placed as an exceptional configuration that the coordinate ($i,j$) of marked vertices on $N\times N$ grid satisfies $j=i$ and $j=i+1$, and $N=3, 16$. We compare the success probability of algorithms $\mathcal{A}$, $\mathcal{B}$ and $\mathcal{C}$.

\begin{figure}[htb]
 \centering
 \subfigure[]{\label{sub1c}
 \includegraphics[width=4.2cm]{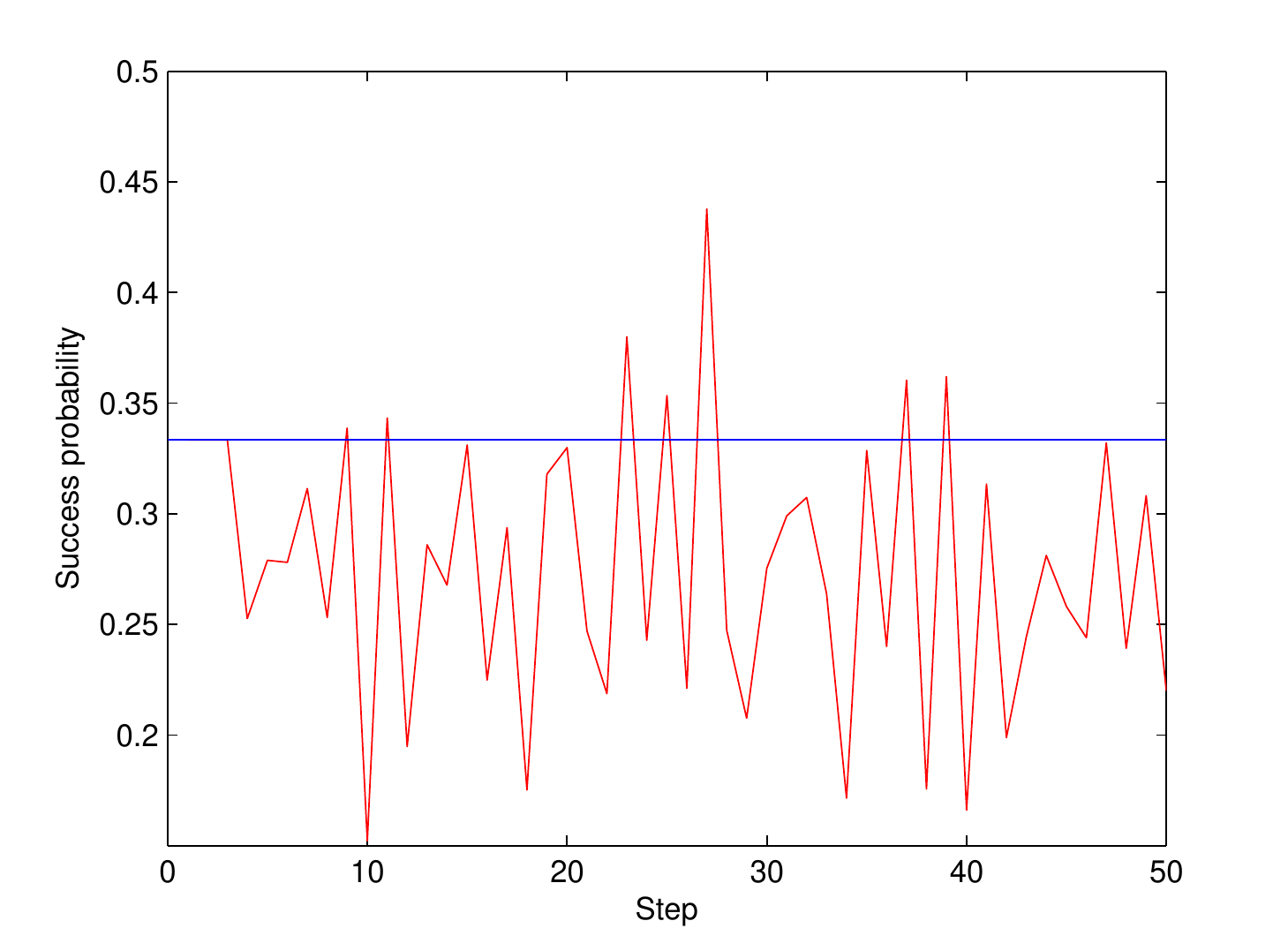}}
 \subfigure[]{\label{sub2c}
 \includegraphics[width=4.2cm]{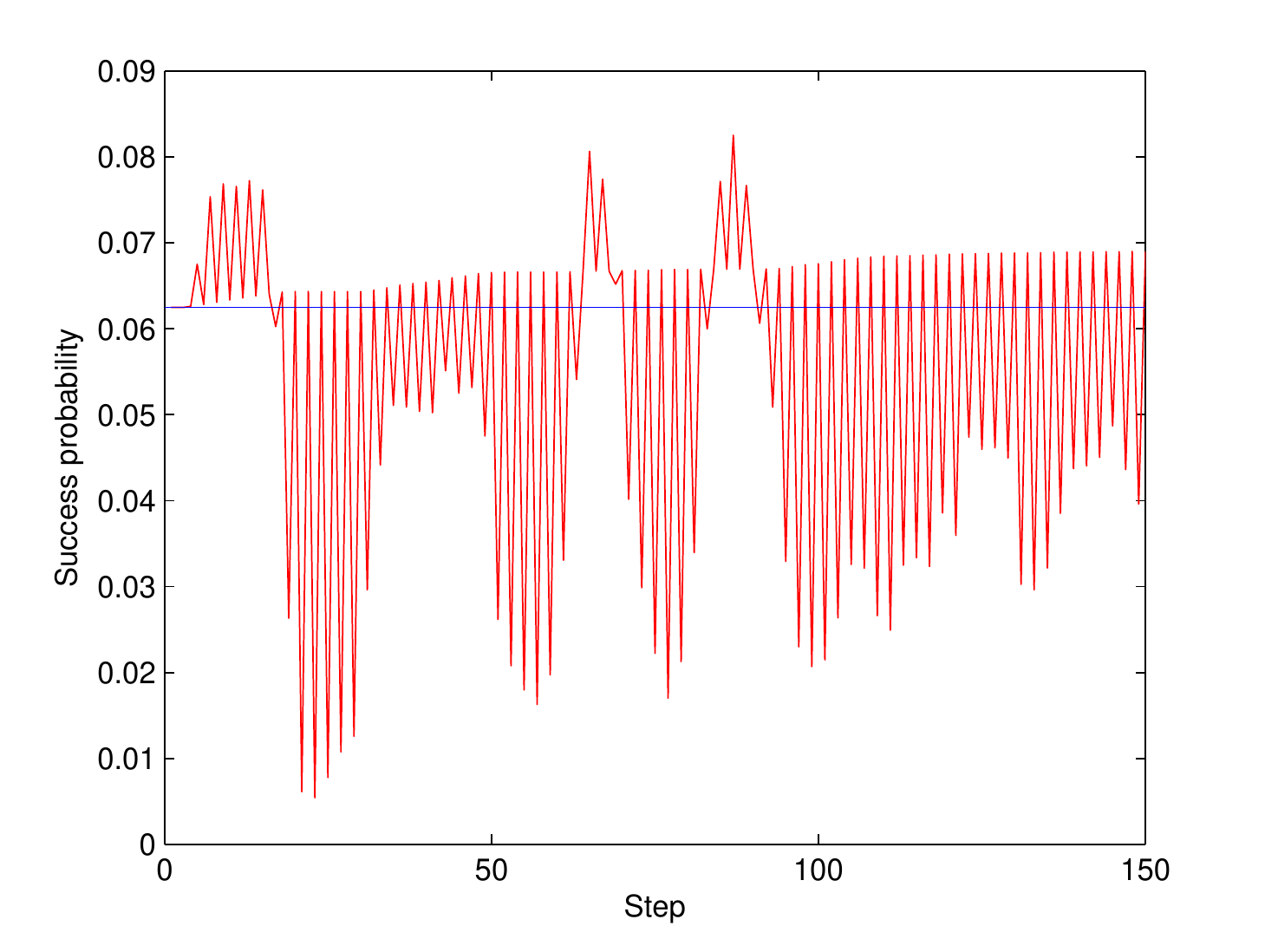}}
  \subfigure[]{\label{sub3c}
 \includegraphics[width=4.2cm]{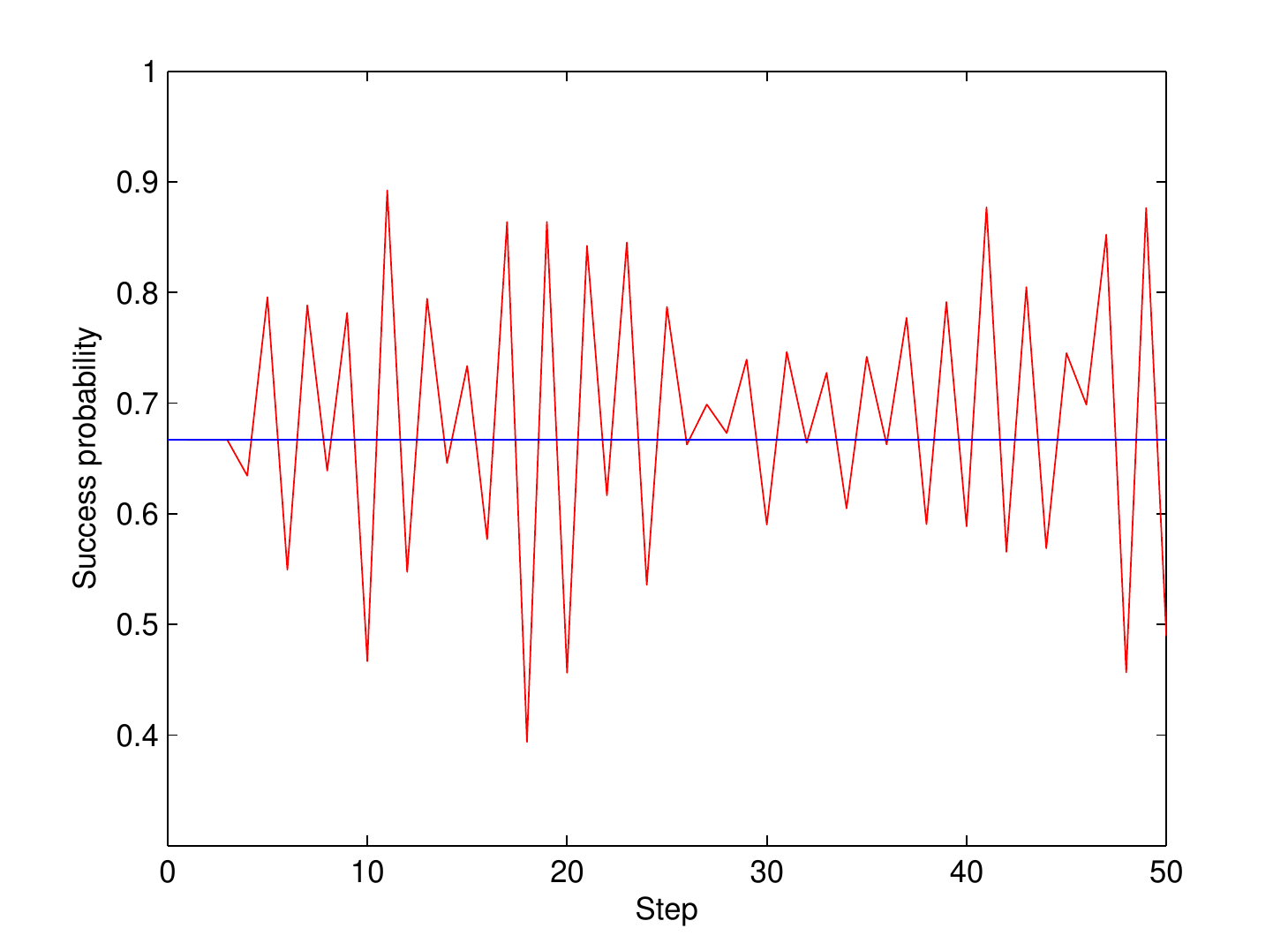}}
  \subfigure[]{\label{sub4c}
 \includegraphics[width=4.2cm]{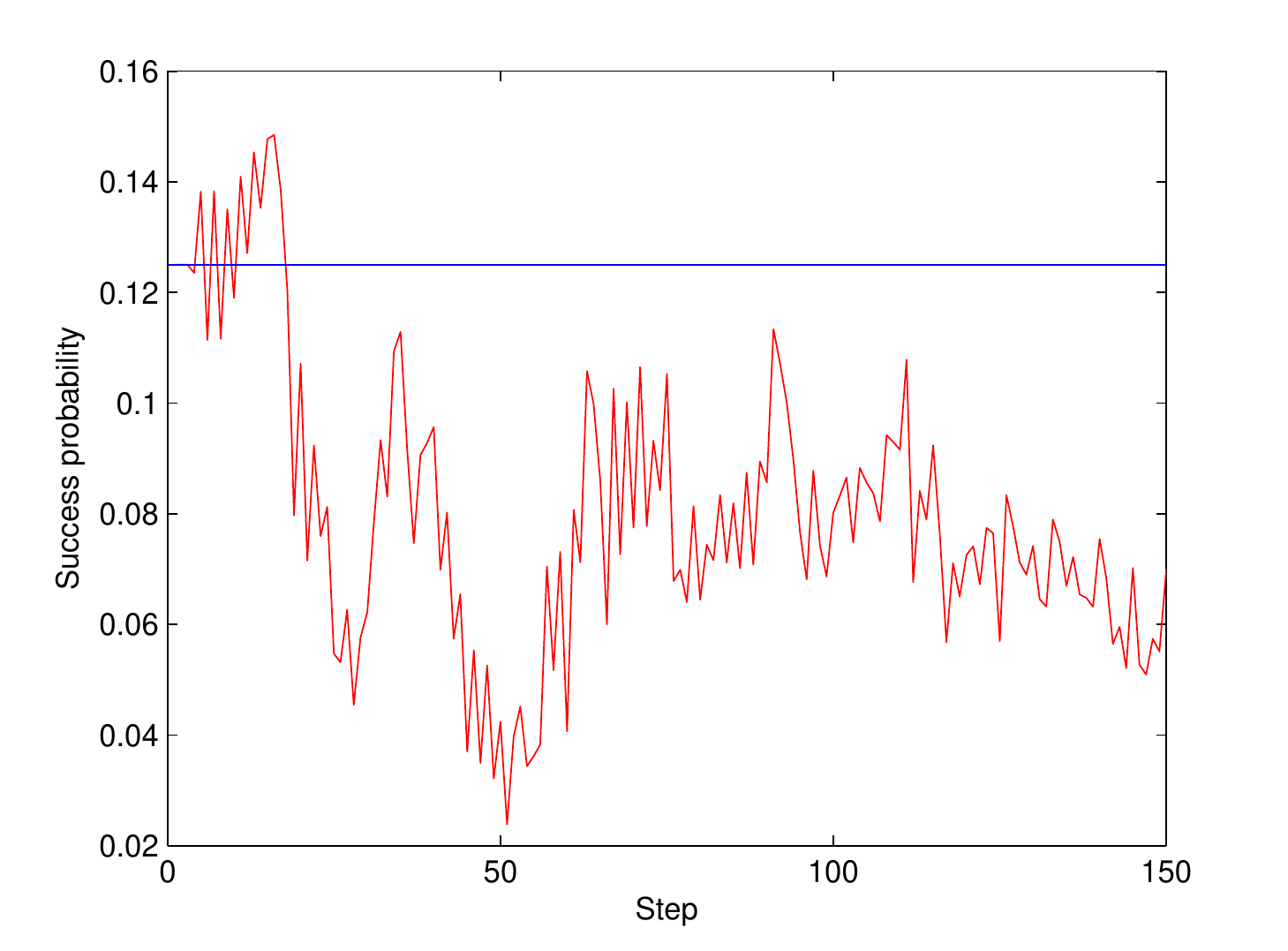}}
 \renewcommand{\figurename}{Figure}
 \caption{The success probability of finding the marked vertices. (a)(c): $3\times 3$ grid; (b)(d): $16\times 16$ grid. (a)-(b): $j-i=0$; (c)-(d): $j-i=0, 1$. The blue line corresponds to algorithm $\mathcal{A}$ and the red curve corresponds to algorithm $\mathcal{C}$.}
 \label{gifgrid}
 \end{figure}

According to Figure~\ref{gifgrid}, the success probability becomes shocked by using algorithm $\mathcal{C}$, compared with algorithm $\mathcal{A}$. For the algorithm $\mathcal{C}$, the success probability reaches a larger value.
The case when $j-i=0$ are shown in Figure \ref{sub1c} and \ref{sub2c}.
On the $3\times 3$ grid, for example, the success probability is about 0.45, and the success probability can reach 0.4377 after 27 steps, whereas the initial probability is 0.33. For the $16\times 16$ grid, the success probability exceeds 0.0800, and the success probability reaches 0.0825 after 87 steps, whereas the initial probability is 0.0625.
In addition, it can be seen from the Figure \ref{sub3c} and \ref{sub4c} that the success probability of algorithm $\mathcal{C}$ has also been improved when $j-i=0$ and $j-i=1$ in comparison to algorithm $\mathcal{A}$.
From this perspective, algorithm $\mathcal{C}$ is better than algorithm $\mathcal{A}$ for this exceptional configuration.  Although algorithm $\mathcal{C}$ does not do better than algorithm $\mathcal{B}$, it is still worth to study. It has already been found that the model has many applications~\cite{wang2017teleportation, shang2019st, shang2019experiment}. So it is an open problem for us to find more applications of this quantum walk model in designing algorithm.

\section{Quantum walk search on the cycle}
Quantum walk can take place on various graph to complete spatial search, in which the cycle is the most common and basic building block in quantum network and have spatial symmetries.
Thus quantum walk on cycle has also been studied extensively \cite{bednarska2003pla, solenov2006pra, sadowski2016jpa}.
Naturally, it is interesting to consider the exceptional configuration and generalized exceptional configuration of the search problem on one-dimensional $N$-cycle with one marked vertex.
In the following discussion, we use the flip-flop shift operator as $S$. Of course, if we use the moving shift operator, the conclusion still holds using the similar proof and thus we will not repeat.

\subsection{\texorpdfstring{$(X,Q)$-}\ Type quantum walk search on the \texorpdfstring{$N$-}\ cycle}
Here, we mainly consider $(G,-)$-type coined quantum walk, which is the most nature and popular. Specifically, it is $(X, Q)$-Type when the dimension of coin Hilbert space is 2, where $Q$ is an arbitrary coin operator and $X$ is the second-order Grover's diffusion transformation, i.e., Pauli matrix.
We find that one dimensional $N$-cycle with one marked vertex is a generalized exceptional configuration under the $(X,Q)$-Type quantum walk search.
We show it in the following theorem.

\begin{theorem}
For an arbitrary non-negative integer $t$, if we run the $(X,Q)$-Type quantum walk search on the $N$-cycle with one marked vertex for $t$ steps, the probability of finding a marked vertex is $\frac{1}{N}$.
\end{theorem}\label{XQtheorem}

\begin{proof}
Without loss of generality, we could assume that the marked vertex is $|0\rangle$. Under the framework of $(X,Q)$-Type coined quantum walk, the flip-flop shift operator is Eq.~(\ref{flipflop}).
And the corresponding coin transformation is
 \begin{equation}
 C=X\otimes(I_{N}-|0\rangle\langle0|)+Q\otimes|0\rangle\langle0|.
 \end{equation}
 So the unitary operator is
\begin{equation}
  \begin{split}
U&=S\cdot C
 =|1\rangle\langle1| \otimes \sum_{k=1}^{N-1}|k-1\rangle\langle k| +(\sqrt{\rho}|1\rangle\langle0|+\sqrt{1-\rho}e^{i\theta}|1\rangle\langle1|)\otimes |N-1\rangle\langle0| \\
 & +|0\rangle\langle0|\otimes\sum_{k=1}^{N-1}|k+1\rangle\langle k|+(\sqrt{1-\rho}e^{i\phi}|0\rangle\langle0|-\sqrt{\rho}e^{i(\theta +\phi)}|0\rangle\langle1|)\otimes|1\rangle\langle0|,
  \end{split}
\end{equation}
it also could be seen as a $2N\times 2N$ matrix:
\begin{equation}\label{XQunitary}
U=
\left(
  \begin{array}{cccccccccc}
    0 & 0 & \cdots & 0 & 1 & 0 & 0 & 0 & \cdots & 0 \\
    \sqrt{1-\rho}e^{i\phi} & 0 & \cdots & 0 & 0 & -\sqrt{\rho}e^{i(\theta+\phi)} & 0 & 0 & \cdots & 0 \\
    0 & 1 & \cdots & 0 & 0 & 0 & 0 & 0 & \cdots & 0 \\
    \vdots & \vdots & \ddots & \vdots & \vdots & \vdots & \vdots & \vdots & \ddots & \vdots \\
    0 & 0 & \cdots & 1 & 0 & 0 & 0 & 0 & \cdots & 0 \\
    0 & 0 & \cdots & 0 & 0 & 0 & 1 & 0 & \cdots & 0 \\
    0 & 0 & \cdots & 0 & 0 & 0 & 0 & 1 & \cdots & 0 \\
    \vdots & \vdots & \ddots & \vdots & \vdots & \vdots & \vdots & \vdots & \ddots & \vdots \\
    0 & 0 & \cdots & 0 & 0 & 0 & 0 & 0 & \cdots & 1 \\
    \sqrt{\rho} & 0 & \cdots & 0 & 0 & \sqrt{1-\rho}e^{i\theta} & 0 & 0 & \cdots & 0
  \end{array}
\right).
\end{equation}
We denote the sum of the i-th row of matrix $U^{t}$ by $U^{t}(i)$, where $t$ is an arbitrary non-negative integer.
Next, we claim that $\left|U^{t}(1)\right|^{2}+\left|U^{t}(N+1)\right|^{2}=\left|U^{t}(N)\right|^{2}+\left|U^{t}(N+2)\right|^{2}
=\left|U^{t}(N-1)\right|^{2}+\left|U^{t}(N+3)\right|^{2}
=\cdots
=\left|U^{t}(3)\right|^{2}+\left|U^{t}(2N-1)\right|^{2}=\left|U^{t}(2)\right|^{2}+\left|U^{t}(2N)\right|^{2}=2$,
and give the proof in detail by induction.

First, when $t=1, 2$, we can get the conclusion by calculation directly.
Assume that the conclusion is established when $t=a$,
\begin{equation}
\left|\sum_{j=1}^{2N}a_{1,j}\right|^{2}+\left|\sum_{j=1}^{2N}a_{N+1,j}\right|^{2}
=\left|\sum_{j=1}^{2N}a_{N,j}\right|^{2}+\left|\sum_{j=1}^{2N}a_{N+2,j}\right|^{2}
=\cdots
=\left|\sum_{j=1}^{2N}a_{2,j}\right|^{2}+\left|\sum_{j=1}^{2N}a_{2N,j}\right|^{2}
=2,
\end{equation}
where $A=U^{a}=(a_{i,j})$ and $U=(u_{i,j})$.
When $t=a+1$, therefore,
\begin{equation}
  \begin{split}
&\left|U^{a+1}(2)\right|^{2}+\left|U^{a+1}(2N)\right|^{2}
=\left|\sum_{j=1}^{2N}\sum_{k=1}^{2N}u_{2,k}a_{k,j}\right|^{2}+\left|\sum_{j=1}^{2N}\sum_{k=1}^{2N}u_{2N,k}a_{k,j}\right|^{2}\\
&=\left|\sqrt{1-\rho}e^{i\phi}(\sum_{j=1}^{2N}a_{1,j})-\sqrt{\rho}e^{i(\theta+\phi)}(\sum_{j=1}^{2N}a_{N+1,j})\right|^2
+\left|\sqrt{\rho}(\sum_{j=1}^{2N}a_{1,j})+\sqrt{1-\rho}e^{i\theta}(\sum_{j=1}^{2N}a_{N+1,j})\right|^{2}\\
&=\left|\sum_{j=1}^{2N}a_{1,j}\right|^{2}+\left|\sum_{j=1}^{2N}a_{N+1,j}\right|^2=2,
  \end{split}
\end{equation}

\begin{equation}
\left|U^{a+1}(3)\right|^{2}+\left|U^{a+1}(2N-1)\right|^{2}=\left|\sum_{j=1}^{2N}a_{2,j}\right|^{2}+\left|\sum_{j=1}^{2N}a_{2N,j}\right|^2=2,
\end{equation}
$$\cdots\cdots$$
\begin{equation}
\left|U^{a+1}(N)\right|^{2}+\left|U^{a+1}(N+2)\right|^{2}=\left|\sum_{j=1}^{2N}a_{N-1,j}\right|^{2}+\left|\sum_{j=1}^{2N}a_{N+3,j}\right|^2=2,
\end{equation}
\begin{equation}
\left|U^{a+1}(1)\right|^{2}+\left|U^{a+1}(N+1)\right|^{2}=\left|\sum_{j=1}^{2N}a_{N,j}\right|^{2}+\left|\sum_{j=1}^{2N}a_{N+2,j}\right|^2=2,
\end{equation}
our assertion is established.

The initial state of particle is
\begin{equation} \label{cycleinitial}
|\varphi(0)\rangle=\frac{1}{\sqrt{2N}}\sum_{d=0}^{1}\sum_{x=0}^{N-1}|d,x\rangle=\frac{1}{\sqrt{2N}}\sum_{x=0}^{N-1}(|0,x\rangle+|1,x\rangle).
\end{equation}
After $t$ steps of this $(X,Q)$-Type coined quantum walk,  the state evolves to
\begin{equation} \label{cyclefinal}
|\varphi(t)\rangle=U^{t}|\varphi(0)\rangle=\frac{1}{\sqrt{2N}}(U^{t}(1),U^{t}(2),\cdots,U^{t}(2N-1),U^{t}(2N))^{T}.
\end{equation}
So the success probability of finding the marked location $|0\rangle$ is
\begin{equation}
\sum_{d=0}^{1}\left|\langle d,0|\varphi(t)\rangle \right|^{2}
 =\left|\langle0,0|U^{t}|\varphi(0)\rangle \right|^{2}+\left| \langle1,0|U^{t}|\varphi(0)\rangle \right|^{2}
 =\frac{1}{2N}(\left|U^{t}(1) \right|^{2}+\left|U^{t}(N+1) \right|^{2})
 =\frac{1}{N}.
\end{equation}

\end{proof}

\begin{remark}
Wong~\cite{ThomasRM17} deduced the conclusion that the two-dimensional grid with a marked diagonal by $(G,-I)$-Type quantum walk can be reduced to a one-dimensional cycle with one marked node by $(X,-I)$-Type quantum walk, and naturally reproduced Ambainis's results~\cite{ambainis2008exceptional} by this kind of reduction.

Similarly, the problem of searching on the two-dimensional grid with a marked diagonal under $(G,F)$-Type quantum walk corresponds to the problem of searching on one-dimensional cycle for single marked vertex by the $(X,H)$-Type quantum walk.
However, the former is not a generalized exceptional configuration through numerical simulation from Figure~\ref{gfgrid}. And the later is a generalized exceptional configuration shown in Theorem 3.
Thus, we cannot obtain the reduction presented by Wong~\cite{ThomasRM17}.
So this kind of higher-dimensional generalization or reduction is related to the choice of coin operator and the class of exceptional configurations.
\end{remark}

However, we will show that the $(X,H)$-Type quantum walk with one marked vertex on one-dimensional cycle is not an exceptional configuration in the following.

\subsection{\texorpdfstring{$(X,H)$-}\ Type quantum walk search on the \texorpdfstring{$N$-}\ cycle}
In order to distinguish exceptional configuration from generalized exceptional configuration, let us consider searching on the $N$-cycle with one marked vertex under $(X,H)$-Type quantum walk, where $H$ represents Hadamard operator.

\begin{theorem}
One-dimensional $N$-cycle with one marked vertex is not an exceptional configuration but a generalized exceptional configuration under $(X,H)$-Type quantum walk search.
\end{theorem}\label{XHtheorem}

\begin{proof}
First, we show that $(X,H)$-Type coined quantum walk search on the $N$-cycle with one marked vertex is not an exceptional configuration.
Without loss of generality, we could assume that the marked vertex is $|0\rangle$.
For the initial state of the particle
\begin{equation} \label{cycleinitialXH}
|\varphi(0)\rangle=\frac{1}{\sqrt{2N}}\sum_{d=0}^{1}\sum_{x=0}^{N-1}|d,x\rangle=\frac{1}{\sqrt{2N}}(|0,0\rangle+|1,0\rangle+\sum_{d=0}^{1}\sum_{x=1}^{N-1}|d,x\rangle),
\end{equation}
after one step $(X,H)$-Type quantum walk, it evolves into
\begin{equation}
|\varphi(1)\rangle=\frac{1}{\sqrt{2N}}(\sqrt{2}|1,N-1\rangle+\sum_{x=1}^{N-1}(|1,x-1\rangle+|0,x+1\rangle)).
\end{equation}
We can see clearly that the system does not just evolve by flipping signs, so it is not an exceptional configuration.

Next, when the parameters $\rho$, $\theta$ and $\phi$ of $Q$ are $\frac{1}{2}$, 0 and 0, the $(X,Q)$-Type coined quantum walk is exactly $(X,H)$-Type. So $(X,H)$-Type coined quantum walk search on the one dimensional $N$-cycle with one marked location is a generalized exceptional configuration.
\end{proof}

\begin{remark}

In addition to $(X,H)$-Type quantum walk search, there are also many examples than can prove exceptional configuration is a true subclass of generalized exceptional configuration.
When we run the $(X,Q)$-Type quantum walk search on the $N$-cycle with one marked vertex, where
the initial state shown in equation (\ref{cycleinitial}) will evolve into
\begin{equation}
|\varphi(1)\rangle=\frac{1}{\sqrt{2N}}((\sqrt{\rho}+\sqrt{1-\rho}e^{i\theta})|1,N-1\rangle
    +(\sqrt{1-\rho}e^{i\phi}-\sqrt{\rho}e^{i(\theta+\phi)})|0,1\rangle
    +\sum_{x=0}^{N-2}|1,x\rangle+\sum_{x=2}^{N}|0,x\rangle)
\end{equation}
after one-step quantum walk.
So as long as $\sqrt{\rho}+\sqrt{1-\rho}e^{i\theta}\neq\pm1$ and $\sqrt{1-\rho}e^{i\phi}-\sqrt{\rho}e^{i(\theta+\phi)}\neq\pm1$, it is not an exceptional configuration but a generalized exceptional configuration.

From the above results, we can see the difference between generalized exceptional configuration and exceptional configuration. In particular,
Theorem 3 does not hold for exceptional configuration in general.
And exceptional configuration is a true subclass of generalized exceptional configuration.
\end{remark}

For the problem of quantum walk-based search algorithm with only one marked vertex, Portugal put forward a method called principle eigenvalue technique in chapter 9 of Ref. \cite{Portugal2018qwsearch}, which can analyze the time complexity effectively.
Unfortunately, it does not fit to solve our problem. There are two reasons.
First, this technique can only give an approximately asymptotic solution. Second, the unitary operator of $(X, Q)$-Type coined quantum walk cannot necessarily be written as the product of original evolutionary operator and the reflection operator.


\section{Dynamics of coherence}
Quantum coherence is an important property of quantum state which has no classical counterpart.
In this section, we mainly consider the dynamics of quantum coherence in the exceptional configuration and generalized exceptional configuration which have been discussed in Theorem 1 and 3.
Since the form of $l_{1}$ norm coherence \cite{baumgratz2014prl} is simple and convenient to calculate, we use $l_{1}$ norm coherence to study the evolution of quantum coherence over time.
The $l_{1}$ norm coherence of a given density operator $\rho$ is described as
\begin{equation}\label{l1norm}
C(\rho)=\sum_{i\neq j}|\rho_{ij}|=\sum_{i\neq j}|\langle i|\rho|j\rangle|.
\end{equation}
The upper bound of the $l_{1}$ norm coherence for any quantum state in a $d$ dimensional system is $d-1$.
Now, we will show the evolution of $l_{1}$ norm coherence about exceptional configuration considered in Theorem 1.

\begin{theorem}
For the exceptional configuration discussed in Theorem 1, the $l_{1}$ norm coherence does not change over time and takes the maximum.
\end{theorem}

\begin{proof}
For the initial state
\begin{equation}
|\psi(0)\rangle=\frac{1}{\sqrt{4N^{2}}}\sum_{d,i,j}|d,i,j\rangle
=\frac{1}{\sqrt{4N^{2}}}(1,1,\cdots,1)^{T},
\end{equation}
the corresponding density matrix is $\rho_{0}=|\psi(0)\rangle\langle\psi(0)|=\frac{1}{4N^{2}}I$.
According to the definition of $l_{1}$ norm coherence shown in Equation \ref{l1norm}, the $l_{1}$ norm coherence of initial state can be calculated as
\begin{equation}
C(\rho_{0})=\sum_{i\neq j}|\langle i|\rho_{0}|j\rangle|=\frac{1}{4N^{2}}4N^{2}(4N^{2}-1)=4N^{2}-1.
\end{equation}
As presented in Theorem 1, the system only evolves by flipping signs. Thus the $l_{1}$ norm coherence of quantum state at any time is the same as $C(\rho_{0})$, which takes the maximum.
\end{proof}

\begin{remark}
Combined with the above analysis, we find that the $l_{1}$ norm coherence does not change over time for any exceptional configuration.
It doesn't depend on what kind of graph we are actually doing the quantum walk on.
Thus, the case of $N\times N$-grid in Theorem 1 is only a special case.
\end{remark}

Next, we will analyze the dynamics of $l_{1}$ norm coherence for generalized exceptional configuration in the case discussed in Theorem 3. For a competitive quantum walk search on the $N$-cycle with one marked vertex, the number of steps required should not exceed $N$. So let us talk about the change in $l_{1}$ norm coherence within $N$ steps.

\begin{theorem}
For the generalized exceptional configuration discussed in Theorem 3, the $l_{1}$ norm coherence of state $|\psi(t)\rangle$ will be as follows when $t<N$:
\begin{equation}
C(\rho_{t})=\frac{[(2N-2t)+t(|p|+|q|)]^{2}}{2N}-1,
\end{equation}
where $p=\sqrt{1-\rho}e^{i\phi}-\sqrt{\rho}e^{i(\theta+\phi)}$ and $q=\sqrt{\rho}+\sqrt{1-\rho}e^{i\theta}$.
\end{theorem}

\begin{proof}
In order to maintain the consistency of symbols, we use the notation in Theorem 3. At time $t$, the state evolves into  $|\varphi(t)\rangle=U^{t}|\varphi(0)\rangle=\frac{1}{\sqrt{2N}}(U^{t}(1),U^{t}(2),\cdots,U^{t}(2N-1),U^{t}(2N))^{T}$, where $U$ is the unitary operator represented in Equation \ref{XQunitary}. Here, we denote $(U^{t}(1),U^{t}(2),\cdots,U^{t}(2N-1),U^{t}(2N))^{T}$ as vector $\overrightarrow{U^{t}}$.
So the density matrix of the whole system at time $t$ can be written as $\rho_{t}=\frac{1}{2N} \overrightarrow{U^{t}}\overrightarrow{U^{t}}^{\dagger}$.

Note that $\overrightarrow{U^{1}}=(1,p,1,\cdots,1,1,\cdots,1,q)^{T}$, $\overrightarrow{U^{2}}=(1,p,p,1,\cdots,1,1,\cdots,1,q,q)^{T}$, 
and so on, we have
\begin{equation}
\begin{split}
\overrightarrow{U^{k+1}}=&(U^{k}(N), \sqrt{1-\rho}e^{i\phi}U^{k}(1)-\sqrt{\rho}e^{i(\theta+\phi)}U^{k}(N+1),
U^{k}(2),U^{k}(3),\cdots,U^{k}(N-1),U^{k}(N+2),U^{k}(N+3),\cdots\\
&\cdots,U^{k}(2N),\sqrt{\rho}U^{k}(1)+\sqrt{1-\rho}e^{i\theta}U^{k}(N+1))^{T}.
\end{split}
\end{equation}
By mathematical induction and simple calculation, we can get
\begin{equation}
U^{t}(1)=U^{t}(N+1)=1, \ \ \ \ \|\overrightarrow{U^{t}}\|_{1}=(2N-2t)+t(|p|+|q|),
\end{equation}
where $t<N$.
Therefore, the $l_{1}$ norm coherence at time $t$ can be calculated as
\begin{equation}
\begin{split}
C(\rho_{t})&=\frac{1}{2N}(\sum_{i,j=1}^{2N}|U^{t}(i)|\cdot|U^{t}(j)|-\sum_{i=1}^{2N}|U^{t}(i)|^{2})\\
&=\frac{1}{2N}(\|\overrightarrow{U^{t}}\|_{1}^{2}-2N)\\
&=\frac{[(2N-2t)+t(|p|+|q|)]^{2}}{2N}-1.
\end{split}
\end{equation}

\end{proof}

\begin{remark}
Different from the exceptional configuration, the quantum coherence will change over time for the generalized exceptional configuration. Thus, in view of quantum coherence, we can distinguish between exceptional configuration and generalized exceptional configuration.
\end{remark}

In conclusion, the $l_{1}$ norm coherence changes over time in the case of generalized exceptional configuration.
And the $l_{1}$ norm coherence in exceptional configuration does not change over time and remains at its maximum.
In term of success probability, the exceptional configuration means that the success probability does not grow over time.
So it is natural for $l_{1}$ norm coherence to be high. Also, this is also consistent with previous coherence results about quantum algorithms \cite{shi2017pra, Liu2019Ent}.

\section{Summary}
Here, we mainly consider the exceptional configurations on two-dimensional grid and generalized exceptional configurations on one-dimensional cycle under the framework of coined quantum walk.

For one thing, we find a wider class of exceptional configurations on $N \times N$ grid for the AKR algorithm. To be more specific, we show that the system only evolves by flipping signs and thus the success probability does not change if the difference or sum between horizontal and vertical coordinates of marked locations equals to a fixed integer. When the difference is equal to zero, it happens to be the known diagonal configuration.
Also, a combination of $j-i=\alpha$ and $j-i=\beta$ would result in a new exceptional configuration when $|\alpha-\beta|=\lfloor\frac{N}{2}\rfloor$, where $(i,j)$ is the coordinate of the marked vertex.
Furthermore, we develop two kinds of modified quantum walks to do search successfully in the exceptional configurations in terms of numerical simulation.

For another thing, we find that the one-dimensional $N$-cycle with one marked vertex is not necessarily an exceptional configuration but a generalized exceptional configuration under $(X,Q)$-Type quantum walk search, where $X$ is Pauli matrix and $Q$ is an arbitrary coin operator.
This warns us not to apply the Grover's diffusion operator on the unmarked vertices if we want to search on one-dimensional cycle for one marked vertex.

Furthermore, we analyze the evolution of $l_{1}$ norm coherence in the two cases mentioned above. We find that, to some extent, coherence can be used as an indicator to distinguish exceptional configuration from generalized exceptional configuration.

This is the first time to come up with the concept ``generalized exceptional configuration" which enlarges the class of exceptional configuration. Taking these results into consideration, we could design more accurate and efficient quantum algorithms in the future.
Certainly, there are many quantum algorithms now. But what problems can they solve effectively?
By the above exceptional configuration or the generalized exceptional configuration, we may provide some division standard in the future. In addition, it is possible for us to find relevant potential applications based on theses fundamental performances which are so different with classical case.

\section{Acknowledgements}
Thanks to the referees for their useful suggestions.
We thank the support of National Key Research and Development Program of China under grant 2016YFB1000902, National Natural Science Foundation of China (Grant No.61472412, 61872352), and Program for Creative Research Group of National Natural Science Foundation of China (Grant No. 61621003).

\end{document}